\documentclass[%
reprint,
superscriptaddress,
nofootinbib,
amsmath,amssymb,
aps,longbibliography
]{revtex4-1}

\newif\ifarxiv % Comment out the next line for journal-style version
\arxivtrue

\usepackage{amsthm,bm,xparse,xcolor,tikz}
\usepgfmodule{decorations}
\usetikzlibrary{decorations,decorations.pathmorphing,patterns}
\usepackage[T1]{fontenc}

\AtBeginDocument{%
    \newwrite\bibnotes
    \def\bibnotesext{Notes.bib}
    \immediate\openout\bibnotes=\jobname\bibnotesext
    \immediate\write\bibnotes{@CONTROL{REVTEX41Control}}
    \immediate\write\bibnotes{@CONTROL{%
    apsrev41Control,author="08",editor="1",pages="1",title="0",year="0"}}
     \if@filesw
     \immediate\write\@auxout{\string\citation{apsrev41Control}}%
    \fi
  }%

\usepackage[nodayofweek]{datetime}  
\usepackage{hyperref}
\usepackage{xcolor}
\usepackage{soul}
\usepackage[capitalise]{cleveref}
\usepackage{enumerate}   
\definecolor{mylinkcolor}{rgb}{0,0,0.8} % set link color here as red,green,blue.
\hypersetup{unicode=true, %
  bookmarksnumbered=false,bookmarksopen=false,bookmarksopenlevel=1, %
  breaklinks=true,pdfborder={0 0 0},colorlinks=true}%
\hypersetup{%
  anchorcolor=mylinkcolor,citecolor=mylinkcolor, %
  filecolor=mylinkcolor,linkcolor=mylinkcolor, %
  menucolor=mylinkcolor,runcolor=mylinkcolor, %
  urlcolor=mylinkcolor}%

\newtheorem{theorem}{Theorem}
\newtheorem{lemma}{Lemma}
\newtheorem{corollary}{Corollary}

\theoremstyle{definition}
\newtheorem{proposition}{Proposition} 
\newtheorem{definition}{Definition}

\newcommand{\ket}[1]{| #1 \rangle}
\newcommand{\bra}[1]{\langle #1 |}
\newcommand{\braket}[2]{\langle #1|#2\rangle}
\newcommand{\ketbra}[2]{|#1\rangle\!\langle#2|}
\newcommand{\id}{\mathbb{I}}
\newcommand{\tr}{{\mathrm tr}}
\newcommand{\ot}{\otimes}
\newcommand{\proj}[1]{\ket{#1}\!\bra{#1}}

\usepackage[most]{tcolorbox}

\newcommand{\dd}{\mathrm{d}} % differential
\newcommand{\error}{\mathrm{err}}

\usepackage[mathscr]{eucal}
\usepackage{circuitikz}
\usepackage{float}
\usepackage[font=small]{caption}

\newtheorem{notation}{Notation}

\usepackage{mdframed}
\def\M{\ensuremath\mathcal}
\def\B{\ensuremath\mathbf}

\usepackage[normalem]{ulem}

\usepackage{subcaption}

\begin{document}

%\title{Entropy Bounds in Two-Way Quantum Cryptographic Protocols via Multiple Hypothesis Testing}

\title{Entropy Bounds via Hypothesis Testing and their
Applications to Two-Way Key Distillation in Quantum Cryptography}

%\title{ Entropy Bounds on Hypothesis Testing and Two-Way Quantum Cryptography }

\author{Rutvij Bhavsar}
\email{rutvij.bhavsar@kcl.ac.uk}

\affiliation{Department of Mathematics, King’s College London,
Strand, London WC2R 2LS, United Kingdom}

\author{Junguk Moon}
\affiliation{School of Electrical Engineering, Korea Advanced Institute of Science and Technology (KAIST), 291 Daehak-ro, Yuseong-gu, Daejeon 34141, Republic of Korea}

\author{Joonwoo Bae}
\email{ joonwoo.bae@kaist.ac.kr }
\affiliation{School of Electrical Engineering, Korea Advanced Institute of Science and Technology (KAIST), 291 Daehak-ro, Yuseong-gu, Daejeon 34141, Republic of Korea}

\date{\today}

\begin{abstract}
Quantum key distribution (QKD) achieves information-theoretic security for a shared secret key, using distributed quantum states. To establish secret bits, two honest parties apply key distillation protocols to the resulting measurement outcomes. In this work, we establish a rigorous connection between the key rate achievable by applying two-way key distillation, such as advantage distillation, and quantum hypothesis testing, via an integral representation of the conditional entropy in terms of the error probabilities of binary state discrimination. This connection significantly improves bounds on key rates at finite blocklengths relative to existing fidelity-based bounds. Moreover, in an important special case, it closes the gap between the known sufficient and conjectured necessary conditions for key generation in the large-blocklength limit. More generally, it establishes a matching asymptotic necessary condition, leaving the corresponding finite-blocklength result open. More broadly, our work shows how results in quantum hypothesis testing can directly sharpen the security analysis of QKD.
\end{abstract}

\maketitle

\setcounter{footnote}{0}

\ifarxiv\section{Introduction}\else \noindent{\textit{Introduction.}--}\fi

%\section{Introduction} 
Quantum key distribution (QKD) enables two distant parties to
achieve information-theoretic security, without computational
assumptions, by exploiting quantum resources. A practical
implementation of QKD may contain noise and device imperfections,
which can be exploited by an eavesdropper and, consequently, reduce
the achievable key rate or ultimately prevent secret key generation.
Two-way key distillation protocols, such as advantage distillation
(AD) \cite{maurer1993secret,bae2007key}, provide a powerful means of
overcoming this limitation, enabling secret key generation even in
the presence of high levels of channel noise. More recently, it has
also been exploited in device-independent (DI) QKD protocols
\cite{tan2020advantage,stasiuk2022quantum,hahn2022fidelity}, where
not only the communication channel but also the states and
measurement devices are treated as untrusted.

Despite these advances, the fundamental limits on when key can be
derived from advantage distillation remain poorly understood. Prior
work \cite{stasiuk2022quantum} identified a sufficient condition
for key generation based on the quantum Chernoff divergence
\cite{audenaert2007discriminating,nussbaum2009chernoff} and
conjectured a matching necessary condition. The emergence of
quantities such as the quantum Chernoff divergence in the security
analysis suggests that the performance of AD may be governed by a
deeper connection with quantum state discrimination, indicating that
techniques developed for the latter may provide a natural framework
for analysing the former.

In this work, we show that the conditional entropy governing the key
rate of advantage distillation protocols admits a natural
representation in terms of a family of binary quantum state
discrimination problems. Building upon a recent integral
representation of the conditional entropy
\cite{frenkel2023integral,jencova}, we reformulate the entropy of a
classical-quantum (cq) state in terms of a family of binary quantum
hypothesis testing problems with arbitrary prior probabilities. This
reformulation allows us to compute the von Neumann entropy through
methods developed in quantum state discrimination, providing a
systematic route for translating well-established techniques from
quantum hypothesis testing into entropy bounds for QKD security
proofs.

The resulting framework offers several advantages. First, it
provides a unified explanation for the appearance of
distinguishability measures such as the fidelity and the quantum
Chernoff bound in the security analysis of advantage distillation,
showing that several previously known entropy bounds arise naturally
as consequences of standard results in quantum hypothesis testing.
Second, it enables more refined hypothesis-testing techniques to be
translated directly into entropy bounds, leading to improved
key-rate bounds. Third, by reformulating entropy estimation as a
multi-copy quantum hypothesis testing problem, well-established
finite-copy discrimination techniques become directly applicable to
the analysis of practical advantage distillation protocols. Here we
demonstrate that this yields significantly tighter entropy bounds
while remaining computationally efficient for blocklengths (the
number of rounds combined to distil each bit of key) that would
otherwise be difficult to compute using exact
techniques. Finally, in the asymptotic regime, we establish a
matching necessary condition that holds in the limit of large
blocklength rather than at any single finite blocklength, and, in an
important special case, strengthen this to show that a positive key
rate cannot persist for arbitrary blocklength if this condition is
not met.

\ifarxiv\section{Advantage distillation protocol}\else \noindent{\textit{Advantage distillation protocol.}--}\fi \label{sec:advantange_distillation}
 
Here we consider the repetition-code advantage distillation protocol introduced in Ref. \cite{tan2020advantage}. In each round $i \in \{ 1 , 2 , \cdots , N \}$, Alice and Bob receive inputs $X_{i} = x$ and $Y_{i} =y$, respectively. The outputs of the devices are labelled $A_{i} = a$ (corresponding to POVMs $\{ M_{a|x} \}$) and $B_{i} = b$ (corresponding to POVMs $\{ N_{b|y} \}$). Alice and Bob use a subset of their data to estimate the input-output statistics (parameter estimation) and then perform classical post-processing (using two-way communication) to generate the final key (unless the protocol aborts). Only rounds with $X=Y=0$ are used for key generation and the remaining rounds are used to establish additional statistics (if required by the protocol) that determine whether the protocol aborts. Following \cite{tan2020advantage}, Alice  then generates a random bit $T_{i}$ and announces it publicly. Both parties then modify their raw bits $A_{i}$ , $B_{i}$ to $\tilde{A}_{i} = T_{i} \oplus A_{i}$ and $\tilde{B}_{i} = T_{i} \oplus B_{i}$. The statistics $p_{\tilde{A} \tilde{B}|00}$ takes the form: 
\begin{eqnarray*}
    p_{\tilde{A} \tilde{B}|X=0,Y=0}(ab) = \left( \left(1 - \epsilon \right)    \delta_{a\oplus b,0} +  \epsilon \delta_{a \oplus b , 1} \right)/2 , 
\end{eqnarray*}
where $\delta_{i,j}$ denotes the Kronecker delta function and $\epsilon \in (0,  1/2)$. The classical quantum (cq) state of a representative single round of the protocol is then given by
\begin{eqnarray}
&&\sum p_{\tilde{A} \tilde{B} | X = 0 , Y = 0}(ab) \proj{ab} \ot \rho_{ET |ab}, ~~\mathrm{where} \nonumber \\
&& \rho_{ET| ab} =
 \frac{1}{\bar{N}}\sum_{i = 0 , 1}  {\tr}_{AB}(M_{a \oplus i|0} \ot  N_{b \oplus i|0}  \rho_{Q_{A} Q_{B} E}) \ot \proj{i}_{T}  ~~~~~\nonumber %\label{eqn: post-measurement_state} 
\end{eqnarray}
and $\bar{N} = 2 p_{\tilde{A} \tilde{B} | X = 0 , Y = 0}(ab) $ \footnote{It is worth noting that although the public symmetrisation generally produces mixed conditional states on ET, the coherent symmetrisation construction of Ref. \cite{TanDI} provides an equivalent realisation with pure conditional states when Alice and Bob share two-qubit states and perform rank one projective measurements. See the supplementary material for more details.}. Formally $\rho_{ET| ab} $ is the eavesdropper's single-round state when $\tilde{A} = a$ and $\tilde{B} = b $ in the rounds when $X = Y = 0$ , where $\rho_{Q_{A} Q_{B} E}$ is the initial tripartite state shared by Alice, Bob and Eve. 

In the advantage-distillation step, the key-generating rounds are divided into blocks of $n$ rounds (i.e. $n$ is the blocklength for the protocol). For each block, Alice chooses a uniform bit $C$ and sends $\mathbf{M} = \tilde{\mathbf{A}} \oplus (C,\dots,C)$, where $\tilde{\mathbf{A}}$ is her string of outputs. Bob accepts the block if $\tilde{\mathbf{B}}\oplus \mathbf{M} = (C',\dots,C')$ and communicates $D=1$; otherwise he rejects the block ($D=0$)~\cite{tan2020advantage,bae2007key}.

Here we assume the i.i.d. structure for states and measurements and consider security against collective attacks as is also considered in \cite{tan2020advantage} (see also \cite{bae2007key}). Under such a security model, the condition to derive key  for the protocol is given by~\cite{Renner}
\begin{eqnarray}
   H(C|\B{E}\B{T} \B{M}; D = 1) - H(C|C';D = 1) > 0.
\end{eqnarray}
Importantly \cite{tan2020advantage} provides a sufficient condition for distilling one bit of key in terms of the fidelity. In particular, they show that it is possible to extract a key via the advantage distillation protocol if 
\begin{equation}
 F(\rho_{ET|00}, \rho_{ET|11})^2 > \beta_{\epsilon}:=  \frac{\epsilon}{1 - \epsilon}.
\end{equation}
where \(F(\tau, \sigma)^2 = \| \sqrt{\tau} \sqrt{\sigma}\|_1^2\) is the squared fidelity.

The possibility for a tighter sufficient condition was proposed \cite{hahn2022fidelity} in terms of the (non-logarithmic) Chernoff divergence, which was later established rigorously in~\cite{stasiuk2022quantum}. Their condition states that secret bits can be extracted, i.e.,
a key rate is positive, if
\begin{equation}
Q(\rho_{ET|00}, \rho_{ET|11}) >\beta_{\epsilon},
\end{equation}
where the Chernoff bound \(Q(\tau, \sigma)\) is defined by
\begin{equation}
Q(\tau, \sigma) := \inf_{s \in [0,1]} \tr \left( \tau^s \sigma^{1-s} \right).
\end{equation}
Since \(Q(\tau, \sigma) \geq F(\tau, \sigma)^2\) for any pair of quantum states~\cite{iten2020relations}, the Chernoff bound condition is never weaker than the fidelity condition; however, the fidelity may still be easier to compute in practice for certain states.

\ifarxiv\section{Entropy Bounds from Binary Hypothesis Testing}\else \noindent{\textit{Entropy Bound from Binary Hypothesis Testing}--}\fi \label{sec: measSetup}

Our starting point is the following integral representation of the  relative entropy,
\begin{equation}\label{eqn: integral_representation_main}
    D(\tau \,\|\, \sigma) = \frac{1}{\ln2} \left( \tr(\tau - \sigma) + \int_{-\infty}^{\infty} \frac{\mathrm{d}t}{|t|(t-1)^2} \, \tr^{-} A(t) \right),
\end{equation}
where \( A(t) = (1 - t) \tau + t \sigma \), and \( \tr^{\pm}(\cdot) \) denotes the trace of the positive (respectively negative) part of a Hermitian operator and $\ln$ denotes the natural logarithm. This representation was originally derived in~\cite{frenkel2023integral}, and subsequently refined into the form used in this work in ~\cite{jencova}. Furthermore, integral representations of entropy have already proven useful in deriving lower bounds on key rates and randomness generation rates in quantum cryptographic protocols. In particular, the recent works~\cite{kossmann2024bounding, kossmann2024optimising} employ~\eqref{eqn: integral_representation_main} to obtain lower bounds on global and local randomness rates for standard DI randomness generation and QKD protocols. 

The goal is to establish a direct connection between the von Neumann entropy of classical–quantum (cq) states and a corresponding quantum state discrimination task. We exploit this link in the context of advantage distillation protocols, thereby bridging the QKD literature with the rich framework of asymptotic quantum hypothesis testing. This perspective not only provides lower bounds on achievable key rates but also enables upper bounds by drawing on established results from hypothesis testing. 

Concretely, we focus on the von Neumann entropy of cq states of the form 
\(\rho = \sum_{c} p_c \proj{c} \otimes \omega_c\), 
where \(\{ \omega_c \}_{c}\) are quantum states. The conditional entropy \(H(C|E)\) of a binary cq state can be represented as an integral over a family of binary hypothesis testing problems between the states \(\{ \omega_c \}\) with varying prior probabilities.

Formally, a state discrimination problem is defined by an ensemble \(\mathcal{E} = \{ \pi_i \omega_i \}_{i=1}^n\), where each state \(\omega_i\) occurs with prior probability \(\pi_i\). The goal is to construct a measurement that maximises the probability of correctly identifying the state. The minimal error probability of distinguishing an ensemble $\M{E}$ in quantum theory is:
\begin{equation}
    p_{\error}(\mathcal{E}) := 1 - \max \left(   \sum_i \pi_i \, \tr \left[ \Lambda_i \omega_i \right] \right),
\end{equation}
where the maximum is taken over all POVMs \(\{ \Lambda_i \}_{i =1}^{n}\) (see~\cite{KRS}; see also~\cite{bae2015quantum} for a comprehensive review of quantum state discrimination in quantum information theory).

The central tool that allows us to connect the conditional entropy $H(C|E)$, which determines the key rate, with quantum hypothesis testing is the following integral representation:\\

\begin{proposition}\label{prop: integral_representation_main_text}
Let $\omega_0, \omega_1$ be states on $\mathcal{H}_E$, and let
$\rho_{CE} = \frac{1}{2}\sum_{c \in \{0,1\}} \proj{c}_C \otimes \omega_c$.
For $s \in [0,1]$, let $p_{\error}(\mathcal{E}(s))$ denote the minimal error
probability of discriminating the ensemble
$\mathcal{E}(s) = \{ s\,\omega_0,\, (1-s)\,\omega_1 \}$. Then
\begin{equation}\label{eqn: intergral_main_text}
H(C|E)_{\rho} = \int_0^1 \frac{p_{\error}(\mathcal{E}(s))
+ p_{\error}(\mathcal{E}(1-s))}{2s\,\ln 2}\, ds .
\end{equation}
\end{proposition}

A more general expression of \eqref{eqn: intergral_main_text} for arbitrary priors $p_{c}$ is derived in Supplementary Material (Section  \ref{app: integral representation derivation}). 

An immediate application of the integral representation is that it naturally yields a variety of lower bounds on the conditional von Neumann entropy in terms of quantities that arise in quantum state discrimination. For instance, consider the well-known textbook lower bound on $p_{\error}(\M{E})$ in terms of fidelity of the states~\cite{wilde2013quantum}:
\begin{equation}\label{eqn: error_probability_fidelity}
   p_{\error}(\mathcal{E}(s)) \geq \frac{1}{2}\left( 1 - \sqrt{1 - 4s(1-s) F(\omega_0, \omega_1)^2} \right) .
\end{equation}
Substituting this into the integral expression for \(H(C|E)_\rho\) immediately yields the bound
\begin{equation}\label{eqn: Fidelity_lower_bound}
    H(C|E)_\rho \geq 1 - h_2\left( \frac{1}{2} + \frac{1}{2} F(\omega_0, \omega_1) \right),
\end{equation}
where \( h_{2}(.) \) denotes the binary entropy function. This reproduces the well-known result from~\cite{roga2010universal}, which was also used by Tan \textit{et al.}~\cite{tan2020advantage} for their security analysis.

Similarly, one may employ a bound on the error probability of the ensemble \(\mathcal{E}(s)\) in terms of the error probability at \(s=1/2\):
\begin{equation}\label{eqn: error_probability_symmetry}
    p_{\error}(\mathcal{E}(s)) \geq 2 \min\{s, 1-s\} \, p_{\error}(\mathcal{E}(1/2)). 
\end{equation}
This gives the bound
\begin{eqnarray}\label{eqn: min_entropy_bound}
 H(C|E)_{\rho} \geq 2 (1 - 2^{-H_{\min}(C|E)_{\rho}}) 
\end{eqnarray}
which corresponds to the result of~\cite{briet2009properties}, and was subsequently applied by Stasiuk \textit{et al.}~\cite{stasiuk2022quantum} in their work relating advantage distillation to the quantum Chernoff bound.

The derivation of \eqref{eqn: min_entropy_bound} is given in Supplementary Material (Section \ref{app: upper and lower bounds}). It is worth noting that these known lower bounds emerge naturally from Proposition~\ref{prop: integral_representation_main_text} by applying standard bounds on hypothesis testing, as presented in introductory texts on quantum information.

Furthermore, the integral expression in Proposition~\ref{prop: integral_representation_main_text} can yield a tighter bound on the conditional entropy than either of the individual bounds given in~\eqref{eqn: Fidelity_lower_bound} and~\eqref{eqn: min_entropy_bound}. In particular, by combining the previously established error probability bounds~\eqref{eqn: error_probability_fidelity} and~\eqref{eqn: error_probability_symmetry} and taking their pointwise maximum over $s\in[0,1]$, we obtain a bound on $p_{\error}(\mathcal{E}(s))$ that is at least as strong as either individual bound. Substituting this improved bound into the integral then yields a tighter lower bound on $H(C|E)_\rho$ whenever both the fidelity and the error probability at $s=1/2$ are known. The resulting improvement is illustrated in Fig.~\ref{fig:tighter_lower_bound_single_shot}.

It is important to emphasise that our method not only provides lower bounds on the conditional entropy, but also allows us to derive upper bounds. In particular, we can exploit the inequality in the following \cite{audenaert2007discriminating}:
\begin{equation}\label{eqn: operator_ineqality}
     \tr(A^{\alpha} B^{1-\alpha}) \geq  \frac{1}{2} \tr(A + B - |A - B|),
\end{equation}
which holds for all  \( A, B \succeq 0\) and for all \( \alpha \in [0,1] \). As shown in \cite{audenaert2007discriminating}, Eq. \eqref{eqn: operator_ineqality} can be used to get upper bounds on the error probability. 

Substituting the resulting upper bound into the integral representation of the entropy leads to the following family of upper bounds for $\alpha \in (0, 1)$:
\begin{equation}
    H(C|E)_\rho \leq \frac{\pi}{2 \ln 2} \left( \frac{Q_{\alpha}( \omega_0 \|  \omega_1)}{\sin(\pi \alpha)} \right),
\end{equation}
where \( Q_{\alpha}(\tau \| \sigma) = \tr(\tau^{\alpha} \sigma^{1 - \alpha}) \) is the Petz--Rényi quantity (See Supplementary Material for the precise definition) and $\alpha \in (0 , 1)$. We also formally derive this bound in Supplementary Material (Section \ref{app: upper and lower bounds}). 
 
\begin{figure}[t]
\captionsetup{font=footnotesize}
\centering
\includegraphics[width=0.9\columnwidth]{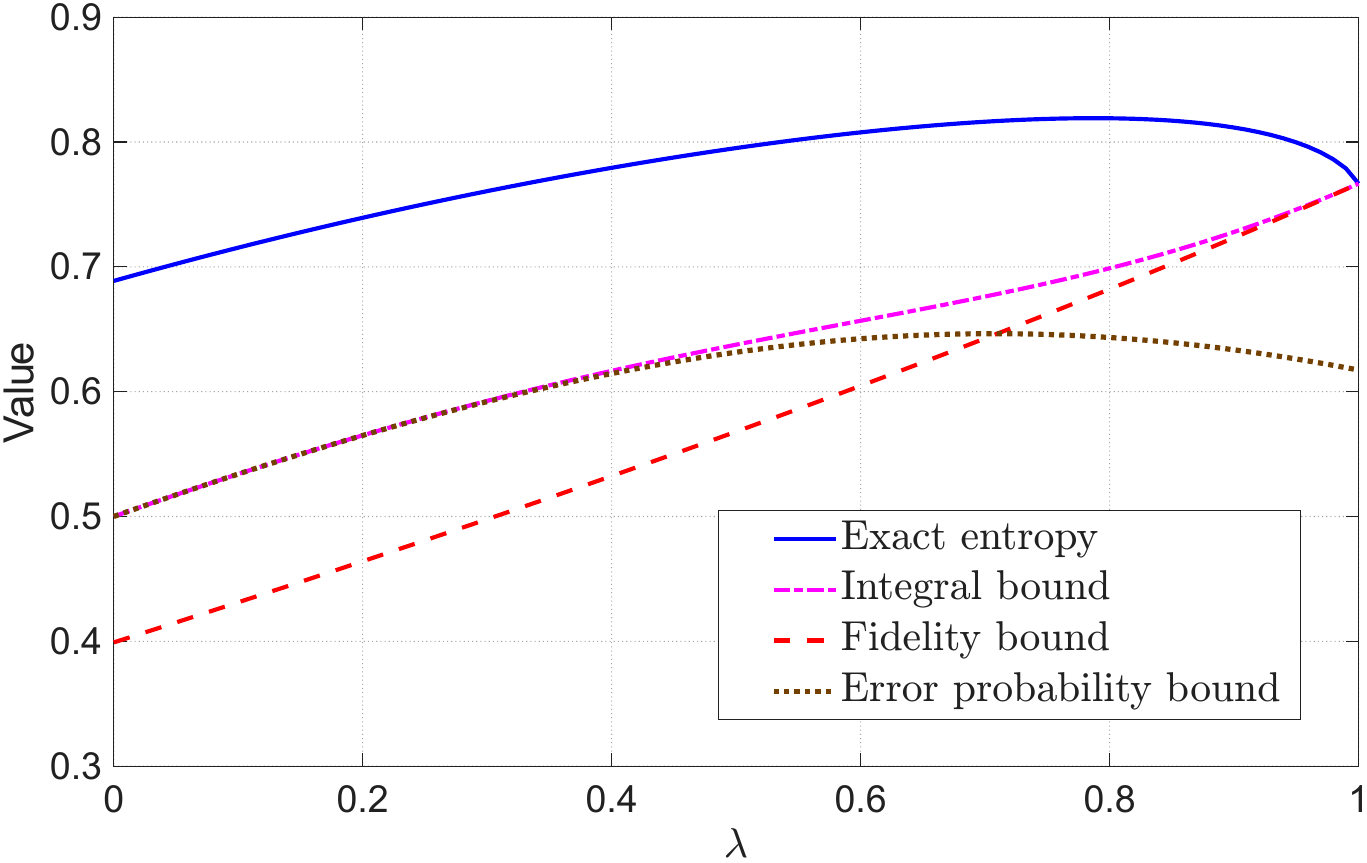} 
\caption{
Different lower bounds on \( H(C|E)_{\rho} \) as a function of the parameter \( \lambda \), where \( \rho = \frac{1}{2} \sum_{i} \proj{i} \otimes \tau_i \). The state \( \tau_0 \) is fixed as the qubit with Bloch vector \( (0, 0, 1) \), while \( \tau_1 \) has Bloch vector \( \left( \lambda/\sqrt{2}, 0, \lambda/\sqrt{2 } \right) \) with \( \lambda \in [0, 1] \). The integral bound is obtained by substituting into Eq. \eqref{eqn: intergral_main_text} the pointwise maximum of the error-probability lower bounds in Eqs. \eqref{eqn: error_probability_fidelity} and \eqref{eqn: error_probability_symmetry}. The fidelity bound is computed using \eqref{eqn: error_probability_fidelity} and the error probability bound corresponds to the one obtained via \eqref{eqn: error_probability_symmetry}.}
\label{fig:tighter_lower_bound_single_shot}
\end{figure}

\ifarxiv\section{Applications to the QKD protocol}\else \noindent{\textit{Application of the integral representation to AD-based QKD protocol.}--}\fi \label{sec:introPro}

The key rate for the QKD protocol under the security model as described in Section \ref{sec:advantange_distillation} (~\cite{stasiuk2022quantum})  is given by $(H(C|E)_{\rho^{(n)}} - h_2(\delta_n)) \cdot ((1 - \epsilon)^n + \epsilon^n)/n$. Here $\rho^{(n)}  = \sum_{c \in \{0,1\} } 1/2 \proj{c} \ot \omega_{c}^{(n)} $ with (see also Supplementary Material (Section \ref{app: integral representation QKD}))
\[
\omega_i^{(n)} = (1 - \delta_n)\, \rho_{ET|ii}^{\otimes n}
                 + \delta_n\, \rho_{ET|i\bar{i}}^{\otimes n}, 
\]
$\bar{i} = i \oplus 1$ and $\delta_n := \epsilon^n/( \epsilon^n + (1 - \epsilon)^n )$.
Therefore, to bound the entropy $H(C|E)_{\rho^{(n)}}$ (and hence the key rate), our  goal is to bound the ensemble error probability
$p_{\error}(\mathcal{E}^{(n)}(s))$, where
$\mathcal{E}^{(n)}(s) = \{ s\,\omega_0^{(n)},\,(1-s)\,\omega_1^{(n)} \}$, and thereby
compute the conditional entropy via Eq.~\eqref{eqn: intergral_main_text}. 

To compute a bound on $p_{\error}(\mathcal{E}^{(n)}(s))$, we first lower bound it by
\[
(1 - \delta_n)\, p_{\error}(\mathcal{E}_0^{(n)}(s))
  + \delta_n\, p_{\error}(\mathcal{E}_1^{(n)}(s)),
\]
where
$\mathcal{E}_0^{(n)}(s)=\{ s\,\rho_{ET|00}^{\otimes n}, (1-s)\,\rho_{ET|11}^{\otimes n}\}$
and
$\mathcal{E}_1^{(n)}(s)=\{ s\,\rho_{ET|01}^{\otimes n}, (1-s)\,\rho_{ET|10}^{\otimes n}\}$. This follows from the joint concavity of the optimal binary error probability in the two hypothesis operators.

We now bound $p_{\error}(\M{E}^{(n)}_{i}(s))$ in terms of error probabilities of suitable classical hypothesis problems. 
Following~\cite{nussbaum2009chernoff,audenaert2012quantum}, the error probability for
distinguishing a quantum ensemble $\{s\,\tau^{\otimes n},(1-s)\,\sigma^{\otimes n}\}$ can be lower bounded by
that for the classical ensemble $\{s\,P^{\otimes n},(1-s)\,R^{\otimes n}\}$ (up to the factor of $1/2$), where
\[
P_0(i,j)=\lambda_{\tau_i}\,|\langle x_i|y_{j}\rangle|^2,
\qquad
P_1(i,j)=\lambda_{\sigma_j}\,|\langle x_{i}|y_{j}\rangle|^2,
\]
and $\tau=\sum_i\lambda_{\tau_i}\proj{x_i}$, $\sigma=\sum_j\lambda_{\sigma_j}\proj{y_j}$ are
their spectral decompositions.

This reduction is powerful because the error probability in a classical hypothesis testing can be computed very efficiently and remains tractable even for large $n$. In contrast, exact entropy computation would require diagonalising matrices of dimension $d^{\,n}$ (where $d$ is the Hilbert-space dimension of $\rho_{ET|ij}$), which scales exponentially with $n$. Using efficient numerical methods to bound error probabilities in classical hypothesis problems allows us to treat very large blocklengths $n \sim 10^{3}$ when $\dim(\rho_{ET|ij}) = 4$ on a standard laptop, whereas exact trace–norm computations already become computationally expensive near $n \approx 7$. We emphasise that this advantage becomes even more pronounced as $d$ increases, yielding nontrivial bounds even at small or moderate blocklengths.

Moreover, by taking the pointwise maximum with the fidelity-based error probability bound Eq.~\eqref{eqn: error_probability_fidelity}, one automatically obtains an entropy bound that is never weaker than the fidelity-based entropy bound Eq.~\eqref{eqn: Fidelity_lower_bound}. This fidelity-based entropy bound is commonly used in practice, particularly for large blocklengths that cannot be handled exactly (see e.g.~\cite{tan2020advantage}). The resulting improvements are illustrated in Fig.~\ref{fig:entropy-bounds}, which shows that our method yields substantially tighter bounds.

\begin{figure}[ht]
    
        \centering
        \includegraphics[width= 0.9\linewidth]{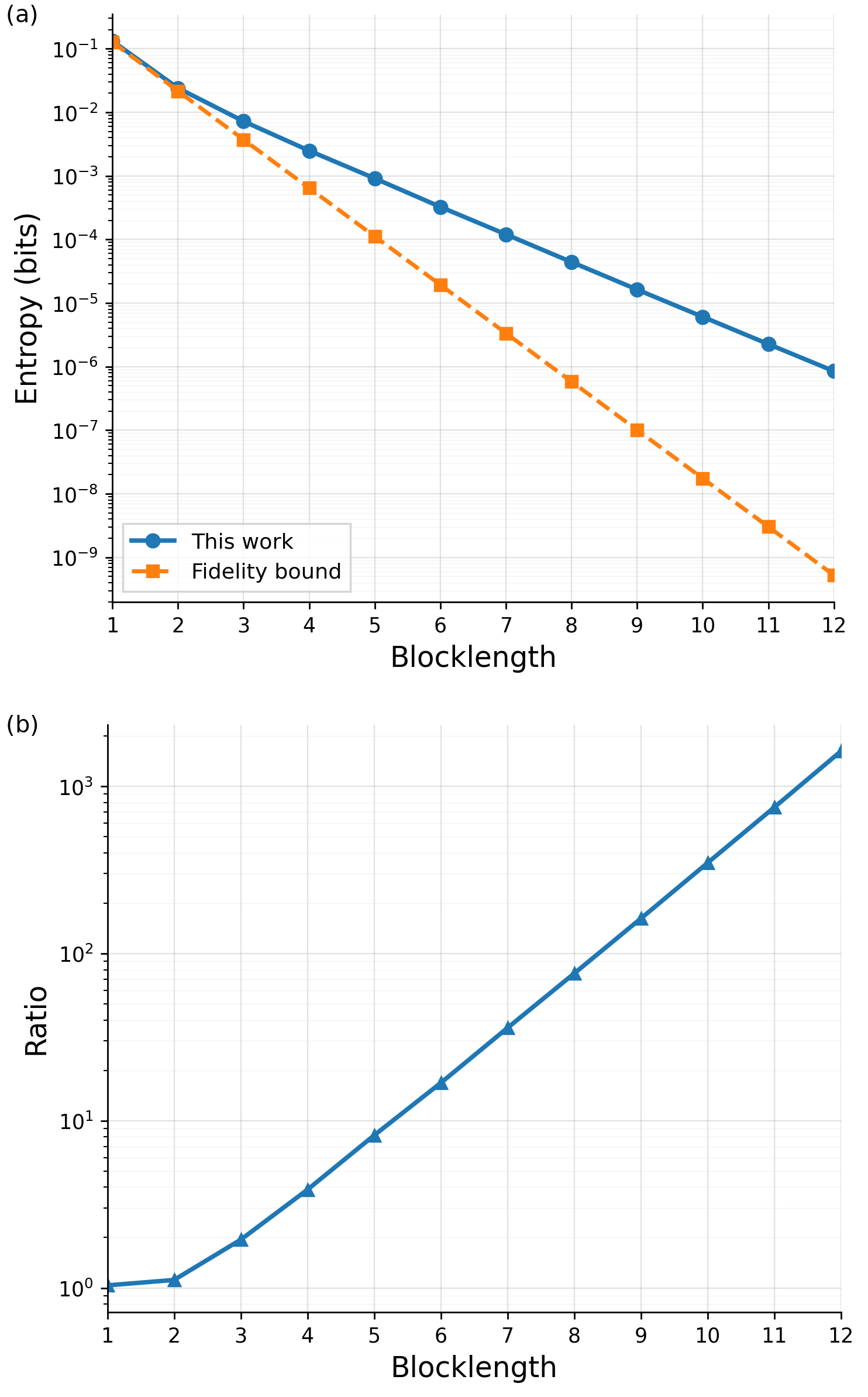}

    \caption{
Comparison of entropy lower bounds for a pair of two-qubit states $\omega_0$ and
$\omega_1$ with fidelity $F(\omega_0,\omega_1)\approx 0.416$.
(a) Lower bound obtained using our method, benchmarked against the fidelity-based
lower bound in Eq.~\eqref{eqn: Fidelity_lower_bound}.
(b) Improvement ratio, defined as the lower bound obtained using our method divided
by the fidelity-based lower bound in Eq.~\eqref{eqn: Fidelity_lower_bound}.
Further details are provided in Sec. \ref{app: figure_data} of the Supplementary Material.}
    \label{fig:entropy-bounds}
\end{figure}

%We now turn to deriving upper and lower bounds on the key rate in the asymptotic limit. The known \textit{sufficient condition}~\cite{tan2020advantage} based on fidelity can be recovered directly from our framework by applying standard fidelity bounds to the corresponding state discrimination task.  Importantly, we are able to extend the results of \cite{stasiuk2022quantum} to derive a lower bound on the conditional entropy (and hence the key rate) in the finite blocklength regime in terms of the Chernoff bound:
%\begin{eqnarray}
% H(C|E)_{\rho} \geq  \frac{  (1 -\delta_{n})  Q_{00;11 }^{n} + \delta_{n} Q_{01;10}^n}{ (n + 1)^{d^2}},
%\end{eqnarray}
%where we use the shorthand notation $Q_{ij;kl} := Q(\rho_{ET|ij}, \rho_{ET|kl})$. This result is obtained by reducing the quantum state discrimination problem to a classical hypothesis testing problem, and then applying the \textit{method of types} and \textit{Sanov's theorem}~\cite{Cover&Thomas} to the classical case.

%However, it is important to note that this bound scales poorly with $n$, as it is a \textit{naïve} lower bound based on the Chernoff bound and leaves scope for tightening the bound using more sophisticated techniques in asymptotic hypothesis testing \cite{audenaert2012quantum,bahadur1960deviations}. We leave such refinements to future work.

%\vspace{0.5em}
We now turn to the \textit{upper bound} on the key rate in the large blocklength limit. In order to show that the key rate is positive, it suffices to show that $\mathscr{R}_{n} > 1$, where 
\begin{eqnarray}
\mathscr{R}_{n}  :=  H(C|E)_{\rho}/h_2(\delta_{n}).
\end{eqnarray}

In Section \ref{app: asymptotic n_and_s_general} of the Supplementary Material, we show that,
\begin{eqnarray}
    \lim_{n \rightarrow \infty} \mathscr{R}_{n} \leq 1
\end{eqnarray}
provided that \( Q(\rho_{ET|00}, \rho_{ET|11}) \leq  \beta_\epsilon \). This shows that the necessary condition for the key rate to be positive,  $\mathscr{R}_{n} > 1$ cannot hold in the asymptotic limit. Note, however, that both the numerator and denominator of this ratio vanish in the asymptotic limit, so the key rate itself tends to zero as $n \rightarrow \infty$. Importantly, we do not rule out that the ratio $\mathscr{R}_{n}$ can approach $1$ from above -- that is, $\mathscr{R}_{n} > 1$ may hold for every finite $n$ -- in which case the key rate remains positive for all finite $n$ even when $Q(\rho_{ET|00} , \rho_{ET|11}) \leq \beta_{\epsilon}$. We expect a more general result with a more refined asymptotic analysis, which we leave for future work.

Here, however, we resolve the issue in an important special case: when the parties share two qubit states and perform rank one projective measurement, we show that an even stronger result holds,
\begin{eqnarray}
    \lim_{n \rightarrow \infty} \mathscr{R}_{n} = 0,
\end{eqnarray}
whenever $Q(\rho_{ET|00}, \rho_{ET|11}) < \beta_{\epsilon}$ and the states $\{ \rho_{ET|ij} \}_{i,j}$ are mutually distinct. The proof of this stronger statement is given in Supplementary Material (Section \ref{app: asymptotic n_and_s_qubit}). Under these assumptions, we further provide a way to bound $Q(\rho_{ET|00}, \rho_{ET|11})$ using only the observed statistics, without requiring any additional characterisation of the devices. Such a bound can also be used to compute the rates for the DIQKD protocol based on AD as shown in \cite{tan2020advantage} (see also  ~\cite{Brown2024deviceindependent, hahn2022fidelity}).

\ifarxiv\section{Discussion}\else \noindent{\textit{Discussion}--}\fi \label{sec:disc}

In this work, we showed that entropy estimation in advantage
distillation admits a reformulation in terms of binary multi-copy
quantum hypothesis testing. Rather than deriving entropy bounds
directly, the relevant conditional entropy is expressed through an
integral representation involving binary discrimination problems
with arbitrary prior probabilities. This reformulation allows the
security analysis of QKD with advantage distillation to draw
directly upon techniques developed for finite-copy and asymptotic
quantum hypothesis testing.

From this viewpoint, several previously known entropy bounds admit a
unified interpretation. In particular, distinguishability measures
such as the quantum Chernoff bound, together with their operational
interpretations, arise naturally from well-known results in quantum
hypothesis testing. In addition, we adapted a standard proof
technique from multi-copy quantum hypothesis testing to obtain new
upper and lower bounds on the conditional von Neumann entropy for
binary cq states. The resulting methods remain computationally
tractable even for blocklengths that are well beyond the reach of
exact numerical evaluation.

The reformulation also provides new insight into the asymptotic
behaviour for the protocol. By expressing the key rate in
terms of asymptotic hypothesis testing, we show that the ratio
$\mathscr{R}_n$ governing whether a positive key rate is achievable
satisfies $\lim_{n\to\infty} \mathscr{R}_n \le 1$ whenever
$Q(\rho_{ET|00},\rho_{ET|11}) \le \beta_\epsilon$, establishing the
necessary condition conjectured in \cite{stasiuk2022quantum} at the
level of this asymptotic ratio. In the special case where Alice and
Bob share qubit states, perform rank one projective measurements, and the
states $\{\rho_{ET|ij}\}_{i,j}$ are mutually distinct, the condition
$Q(\rho_{ET|00},\rho_{ET|11}) < \beta_\epsilon$ gives a much stronger
result: $\lim_{n\to\infty} \mathscr{R}_n = 0$, completing the
asymptotic characterisation of the protocol outright.

Several open directions naturally arise from the reformulation
developed in this work. Some may be addressed using existing
techniques from quantum hypothesis testing, while others raise new
questions within hypothesis testing that are motivated by quantum
cryptography. The immediate challenge is to close the gap between
the necessary and sufficient conditions in the large-but-finite
blocklength regime. This would be resolved either by proving that
there exists a finite $N$ such that $\mathscr{R}_n \leq 1$ for all
$n \geq N$, or by identifying a counterexample for which
$\mathscr{R}_n > 1$ for arbitrarily large values of $n$. Another
direction is to derive tighter finite blocklength entropy bounds in
terms of the quantum Chernoff bound or other appropriate quantity~\cite{audenaert2012quantum}. While the
present work establishes the reformulation and demonstrates its
potential, further refinements of the underlying hypothesis-testing
bounds may substantially improve the resulting entropy bounds.
Combining these bounds with recent work on bounding Petz--R\'enyi
divergences \cite{hahn2024bounds} may offer a path towards improved
estimates for device-independent advantage distillation protocols.
Finally, a natural extension of our work is to generalise the
framework to non-binary input and output alphabets for Alice and
Bob.

\ifarxiv\acknowledgements\else\medskip\noindent{\it Acknowledgements}|\fi
The authors are grateful to Lewis Wooltorton, Seung-hyun Nam and Shashank Ranu for insightful discussions. This work was supported by the Institute for Information \& Communication Technology Promotion (IITP) (RS-2025-02304540, RS-2025-25464876, RS-2025-25464616).  RB was additionally supported by the UK Integrated Quantum Networks Hub (EP/Z533208/1). Generative AI tools were used solely for language editing in the paper, and for assistance with writing codes used to generate figure 2. All scientific content, including all theorems, proofs was produced and verified by the authors, who take full responsibility for the content of this work.

\clearpage
\onecolumngrid   % if the supplement is single-column
\appendix
\counterwithout{equation}{section}
\setcounter{equation}{0}
\setcounter{figure}{0}
\setcounter{section}{0}
\renewcommand{\theequation}{S\arabic{equation}}
\renewcommand{\thefigure}{S\arabic{figure}}
\renewcommand{\thesection}{S\Roman{section}}
\begin{center}\textbf{Supplementary Material}\end{center}

\section{Preliminaries}

Let \( \mathcal{H}_X \) denote the Hilbert space associated with the quantum system \( X \). The set of density operators (i.e., positive semidefinite operators with unit trace) on \( \mathcal{H}_X \) is denoted by \( \mathcal{S}(\mathcal{H}_X) \). By convention $\log$ denotes the base-$2$ logarithm and  $\ln$ denotes the natural logarithm.

A key ingredient in our analysis will be the \emph{Petz--Rényi} and \emph{sandwiched Rényi} quantities and their associated divergences.

\begin{definition}[Petz--Rényi Quantity]
Let \( \tau, \sigma \in \mathcal{S}(\mathcal{H}) \) and $\alpha \geq 0$. The Petz--Rényi quantity is defined as
\[
Q_{\alpha}(\tau || \sigma) := \tr \left( \tau^\alpha \sigma^{1 - \alpha} \right).
\]
\end{definition}
\begin{definition}[Sandwiched Rényi Quantity]
Let \( \tau, \sigma \in \mathcal{S}(\mathcal{H}) \) and $\alpha \geq 0$. The sandwiched Rényi quantity is defined as
\[
\tilde{Q}_{\alpha}(\tau || \sigma) :=  \tr \left[ \left( \sigma^{\frac{1-\alpha}{2\alpha}} \tau \, \sigma^{\frac{1-\alpha}{2\alpha}} \right)^{\alpha} \right].
\]
\end{definition}

Both $Q_{\alpha}$ and $\tilde{Q}_{\alpha}$ are well-defined for \( \alpha < 1 \), or if \( \tau \ll \sigma \) when \( \alpha \geq 1 \). The quantities $\tilde{Q}_{0} , \tilde{Q}_1 , Q_{0} , Q_{1}$ are defined via continuity. These quantities coincide when \( \tau \) and \( \sigma \) commute. Both define divergences known as the \emph{Petz--Rényi divergence} and the \emph{sandwiched Rényi divergence}, respectively.

\begin{definition}[Rényi Divergences]
For any \( \tau, \sigma \in \mathcal{S}(\mathcal{H}) \) and \( \alpha \in (0,1) \cup (1, \infty) \), define
\[
D_{\alpha}(\tau \| \sigma) := \frac{1}{\alpha - 1} \log Q_{\alpha}(\tau ||  \sigma), \quad
\tilde{D}_{\alpha}(\tau \| \sigma) := \frac{1}{\alpha -1 } \log \tilde{Q}_{\alpha}(\tau || \sigma).
\]
If \( Q_{\alpha} \) or \( \tilde{Q}_{\alpha} \) are not well-defined (e.g., \( \alpha > 1 \) and \( \tau \not\ll \sigma \)), the divergences are set to \( \infty \).
\end{definition}

The well-known \emph{quantum relative entropy} (Umegaki divergence) is recovered in the limit \( \alpha \to 1 \):
\[
D(\tau \| \sigma) :=
\begin{cases}
\tr\left[ \tau (\log \tau - \log \sigma) \right], & \text{if } \tau \ll \sigma, \\
\infty, & \text{otherwise}.
\end{cases}
\]
This quantity is the limiting case of both \( D_{\alpha} \) and \( \tilde{D}_{\alpha} \). Additionally, the limit \( \alpha \to \infty \) recovers the max-divergence
\[
D_{\infty}(\tau \| \sigma) := \inf \{ \lambda \in \mathbb{R} : \tau \leq 2^\lambda \sigma \}.
\]
The log-free divergences satisfy the following bounds (see \cite{iten2020relations}):
\[
Q_{\alpha}(\tau \| \sigma) \leq \tilde{Q}_{\alpha}(\tau \| \sigma) \leq Q_{\alpha}(\tau \| \sigma)^\alpha
\]
whenever $\alpha \in (0,1)$.

Another closely related quantity is the \emph{pretty-good fidelity}, defined as
\[
F_{\mathrm{pg}}(\tau, \sigma) := Q_{\frac{1}{2}}(\tau || \sigma) = \tr(\sqrt{\tau} \sqrt{\sigma}),
\]
which is related to the standard fidelity
\[
F(\tau, \sigma) := \tilde{Q}_{\frac{1}{2}}(\tau || \sigma) = \tr\left( \sqrt{\sqrt{\tau} \sigma \sqrt{\tau}} \right).
\]
Throughout, $F$ denotes the root fidelity; hence the conventional squared fidelity is $F(\cdot , \cdot)^2$.

\begin{notation}[Guessing Probability]
Let \( \mathcal{E} = \{ p_X(x) {\omega}_x \} \) be an ensemble. Define the guessing probability as
\[
p_{\mathrm{guess}}\ (X|E)_{\rho} := 2^{- H_{\min}(X|E)_\rho},
\]
where \( \rho = \sum_x p_X(x) \proj{x} \otimes {\omega}_x \) and
\[
H_{\min}(X|E)_\rho := - \inf_{\sigma_E \in \M{S}(\M{H}_{E})} D_{\infty}(\rho_{XE} \| \id_X \otimes \sigma_E)
\]
is the conditional min-entropy. The error probability is then
\[
p_{\error}{(X|E)_{\rho}} := 1 - p_{\mathrm{guess}}{(X|E)_{\rho}}.
\]
\end{notation}
 \noindent\textbf{Note.} The notation used in this document differs slightly from that in the main text. In the main text, we write $p_{\error}(\M{E})$ to denote the optimal error probability of distinguishing an ensemble.  
For every ensemble $\M{E} = \{ p(x)   \omega_{x} \}_{x}$, we associate the state
\[
   \rho = \sum_{x} p(x) \proj{x} \otimes \omega_{x},
\]
and we emphasise that $p_{\error}(\M{E})$ and $p_{\error}(X|E)_{\rho}$ both denote the same quantity: the optimal error probability. In the Supplementary Material, we primarily use the cq state notation $p_{\error}(X|E)_{\rho}$, but occasionally switch to the ensemble notation $p_{\error}(\M{E})$ when more convenient.

When the ensemble contains only two states \( \mathcal{E} = \{ p \tau,  (1-p) \sigma \} \), the optimal error probability is given by the Holevo-Helstrom theorem \cite{Holevo,Helstrom}:
\[
p_{\error}{(X|E)_{\rho}} = \frac{1}{2} - \frac{1}{2} \| p \tau - (1-p) \sigma \|_1,
\]
where \( \| A \|_1 = \tr \left( \sqrt{A^\dagger A} \right) \) is the trace norm. This norm can be computed via the positive-negative decomposition \( A = A_+ - A_- \), where both \( A_+, A_- \geq 0 \), and
\[
\| A \|_1 = \tr(A_+) + \tr(A_-).
\]

For multiple copies of the states \( \mathcal{E} = \left\{ \frac{1}{2} \tau^{\otimes n}, \frac{1}{2} \sigma^{\otimes n} \right\} \), the optimal error probability decays exponentially in terms of the so-called quantum Chernoff bound:
\begin{eqnarray*}
 p_{\error} \sim \exp \left(  n \cdot \ln  Q(\tau, \sigma) \right).   
\end{eqnarray*}
where 
\begin{eqnarray}
  Q(\tau , \sigma)  := \inf_{\alpha \in [0 , 1]} Q_{\alpha}(\tau \| \sigma). 
\end{eqnarray}
Note that notations such as $Q(\tau , \sigma)$ and $F(\tau , \sigma)$,  emphasise the fact that these quantities are symmetric with respect to the arguments. 

\section{Deriving the integral representation of relative entropy}\label{app: integral representation derivation}
The objective of this section is to establish a relationship between the entropy \( H(C | \mathbf{ETM} ; D = 1)_{\rho} \) and the problem of distinguishing the hypotheses \( \omega_0 \) and \( \omega_1 \). The key insight underlying this connection is the integral representation of the relative entropy, given by  
\begin{equation}\label{eqn: integral_representation}
   \ln2 \cdot D(\rho_{AB} || \id_{A} \otimes \rho_{B}) = \tr(\rho_{AB} - \id_{A} \otimes \rho_{B})  + \int_{-\infty}^{0} \dd t  \frac{\tr^{-}((1-t)\rho_{AB} + t \id_{A} \otimes \rho_{B}) }{|t| (|t| + 1)^2} + \int_{1}^{\infty} \dd t \frac{\tr^-((1-t)\rho_{AB} + t \id_{A} \otimes \rho_B)}{t (t -1)^2},
\end{equation}
which was derived in \cite{frenkel2023integral} and further refined in \cite{jencova}.

We note that the integral representation in Eq.~\eqref{eqn: integral_representation_main} has been utilised in recent works \cite{kossmann2024bounding, kossmann2024optimising} in the context of bounding the entropy of classical-quantum (cq) states in quantum key distribution (QKD) and random number generation (RNG) protocols, particularly in device-independent scenarios.  In this work, our focus is to establish a correspondence between computing the key rates for the AD-QKD protocol and a hypothesis testing problem. The error bounds for such hypothesis tests can be computed either analytically or numerically, which, in turn, will yield the desired bounds on the entropy.

Before presenting our main claim, we first present a slightly modified result from \cite{jencova} in order to make its application more suitable for our work.

\begin{lemma}\label{lemm: general_integral}
Let $\rho_{AB} \in \M{S}(\M{H}_{A} \ot \M{H}_{B})$.  Then
\begin{eqnarray}
    D(\rho_{AB} || \id_{A} \ot \rho_{B})  = \frac{1}{\ln 2}\left( \left(1 - d_{A}  \right) + \left( \int_{1}^{\lambda} \frac{\dd s}{s} (1 - d_{A} s)\right) + \int_{0}^{\lambda} \frac{\dd s}{s} \tr^{-}(\rho_{AB} - s \id_{A} \ot \rho_{B}) \right) ,  
\end{eqnarray}
where $d_{A}$ is the dimension of $\M{H}_{A}$ and $\lambda \geq 1$ is any number that satisfies $\lambda \id_{A} \ot \rho_{B} \succeq \rho_{AB}$. 
\end{lemma}
\begin{proof}
Consider the integral representation Eq. \eqref{eqn: integral_representation}. From \cite{jencova}, we can compute the first integral in Eq.\eqref{eqn: integral_representation} as  
\begin{equation}
\int_{-\infty}^{0} \dd t  \frac{\tr^{-}((1-t)\rho_{AB} +t \id_{A} \ot \rho_{B}) }{|t| (|t| + 1)^2} =  \int_{0}^{1} \frac{\dd s}{s} \tr^{-} \left( \rho_{AB} - s \id_{A} \ot \rho_{B} \right).
\end{equation}
We now simplify the second integral 
\begin{equation}
\int_{1}^{\infty} \dd t \frac{\tr^{-}((1 - t) \rho_{AB} + t (\id_{A} \ot \rho_{B})) }{t (t - 1)^2}. 
\end{equation}
Following \cite{jencova}, we know that for $t > 1$, we have that  
$$\tr^{-}\left((1 - t) \rho_{AB} + t \id_{A} \ot \rho_{B} \right) = \tr^{+}\left((t - 1) \left(  \rho_{AB}  - \frac{t}{t - 1}  \id_{A} \ot \rho_{B} \right) \right).$$
Consequently,
\begin{equation}
\int_{1}^{\infty} \dd t \frac{\tr^{-}((1 - t) \rho_{AB} + t (\id_{A} \ot \rho_{B})) }{t (t - 1)^2} = \int_{1}^{\lambda} \frac{\dd s}{s}  \tr^{+}(\rho_{AB} - s \id_{A} \ot \rho_{B}).
\end{equation}
Now we note that $\tr^{+}(\rho_{AB} - s\id_{A} \ot \rho_{B}) - \tr^{-} (\rho_{AB} - s\id_{A} \ot \rho_{B}) = \tr(\rho_{AB} - s\id_{A} \ot \rho_{B}) = 1 - s d_{A}$. This gives 
\begin{eqnarray}
\int_{1}^{\infty} \dd t \frac{\tr^{-}((1 - t) \rho_{AB} + t (\id_{A} \ot \rho_{B})) }{t (t - 1)^2} = \int_{1}^{\lambda} \frac{\dd s}{s} (1 - sd_{A})  + \int_{1}^{\lambda} \frac{\dd s}{s}  \tr^{-}(\rho_{AB} - s \id_{A} \ot \rho_{B}).
\end{eqnarray}
Here, $\lambda \geq 1$ is any number satisfying $\lambda \id_{A} \ot \rho_{B} \succeq \rho_{AB}$. Combining the two results gives:
\begin{eqnarray}
D(\rho_{AB} || \id_{A} \ot \rho_{B}) = \frac{1}{\ln 2}\left( \left(1 - d_{A}  \right) + \left( \int_{1}^{\lambda} \frac{\dd s}{s} (1 - d_{A} s)\right) + \int_{0}^{\lambda} \frac{\dd s}{s} \tr^{-}(\rho_{AB} - s \id_{A} \ot \rho_{B})  \right) 
\end{eqnarray}
\end{proof} 
For applications to the advantage distillation protocol, we only need to consider the case where \( \rho_{AB} \) is a classical-quantum (cq) state. In this case, the problem simplifies significantly. Now we present the main result of this section: 

\begin{theorem}\label{thm: CQ_integral_representation}
Let $\rho_{CE}$ be the classical-quantum (cq) state given by  
\begin{equation}
    \rho_{CE} = \sum_{c \in \{ 0 ,1 \}} p_C(c) \proj{c}_{C} \ot \omega_{c},
\end{equation}  
where $\omega_c$ are normalised quantum states. Define the cq state  
\begin{equation}
\bar{\rho}_{CE}(s) :=  s \proj{0} \ot \omega_0 + (1 -s) \proj{1}\ot \omega_1.
\end{equation}  
Then, the conditional entropy $H(C|E)_{\rho_{CE}}$ has the integral representation  
\begin{equation}
    H(C|E)_{\rho_{CE}} = \int_{0}^{1} \frac{\dd s}{2s \ln 2 } \left( 1 -  || s p_{C}(0) \omega_0   - (1 -s) p_{C}(1) \omega_1 ||_1  - || s p_{C}(1) \omega_{1} - (1 -s) p_{C}(0)\omega_{0} ||_1  \right).  
\end{equation}  
In particular, when $p_{C}(0) = p_{C}(1) = \frac{1}{2}$, then the entropy has the integral representation
% Here, 2 of the denominator of (18) should not be removed
\begin{eqnarray}\label{eqn: integral_representation_of_entropy}
      H(C|E)_{\rho_{CE}} = \int_{0}^{1} \frac{\dd s}{2s \ln 2} \left( p_{\error}(C|E)_{\bar{\rho}_{CE}(s)} + p_{\error}(C|E)_{\bar{\rho}_{CE}(1 -s) } \right).    
\end{eqnarray}
\end{theorem}
\begin{proof}
For any cq state $\rho_{CE}$ we always have $\id_{C} \ot \rho_{E} \succeq \rho_{CE}$. Thus, we can always set $\lambda = 1$ in Lemma \ref{lemm: general_integral} and proceed for this case. Furthermore, we can set $d_{C} =2$ to be the dimension of the classical register $C$. Substituting everything in gives: 
\begin{eqnarray*}  
\ln 2D(\rho_{CE} || \id_{C} \ot \rho_{E}) &=&    (1 -d_{C})  + \int_{0}^{1} \frac{\dd s}{s} \tr^{-} \left(   p_{C}(0)  \proj{0} \ot \omega_{0} + p_{C}(1)\proj{1} \ot \omega_{1} -s \id_{C} \ot \left(p_{C}(0)\omega_0 + p_{C}(1)\omega_1  \right) \right)   \\ 
&=&    (1 -d_{C}) +  \int_{0}^{1} \frac{\dd s}{s} \tr^{-} \left( \left( \proj{0} \ot \left( (1-s) p_{C}(0)\omega_{0} - s p_{C}(1)\omega_{1} \right) \right) + \left( \proj{1} \ot \left(-s p_{C}(0) \omega_{0} + (1-s) p_{C}(1) \omega_{1}\right) \right) \right)      \\ 
&=&  \left(1 - d_{C}  \right) +  \int_{0}^{1}\frac{\dd s}{ s} \left(  \tr^{-} \left( (1-s)p_{C}(0)\omega_{0} - s p_{C}(1)\omega_{1} \right)  + \tr^{-}\left(- sp_{C}(0) \omega_{0} + (1-s)p_{C}(1) \omega_{1} \right) \right)   \label{eqn: D_in_terms_of_positive_and_negative_part}
\end{eqnarray*} 
We can relate $\tr^+(\cdot)$ in terms of the trace norms: 
\begin{eqnarray}
\tr^{+} \left( (s-1)p_{C}(0)\omega_{0} + s p_{C}(1)\omega_{1} \right)  = \frac{(s-1)p_{C}(0) + s p_{C}(1) }{2} + \frac{||s p_{C}(1) \omega_1 - (1 -s) p_{C}(0) \omega_0 ||_1}{2} 
\end{eqnarray}
Similarly, we have that 
\begin{eqnarray}
\tr^{+} \left( (s-1)p_{C}(1)\omega_{1} + s p_{C}(0)\omega_{0} \right)  = \frac{(s-1)p_{C}(1) + s p_{C}(0) }{2} + \frac{||s p_{C}(0) \omega_0 - (1 -s) p_{C}(1) \omega_1 ||_1}{2}. 
\end{eqnarray} 
Therefore, using the fact that $\tr^{+}(- A) = \tr^{-}(A)$ gives the following expression:
\begin{eqnarray*} 
 D(\rho_{CE} || \id_{C} \ot \rho_{E}) &=&    \frac{-1}{\ln 2}  + \int_{0}^{1} \frac{\dd s}{2s \ln2} \left( 2s-1 + ||s p_{C}(0) \omega_0  - (1 -s) p_{C}(1) \omega_1 ||_1 + ||s p_{C}(1) \omega_1 - (1 -s) p_{C}(0) \omega_0 ||_1  \right)    
\end{eqnarray*}
where we have substituted $d_{C} = 2$. Thus, 
\begin{eqnarray}
H(C|E)_{\rho_{CE}} =\int_{0}^{1} \frac{\dd s}{2 s\ln 2 } \left( 1 -  || s p_{C}(0) \omega_0   - (1 -s) p_{C}(1) \omega_1 ||_1  - || s p_{C}(1) \omega_{1} - (1 -s) p_{C}(0)\omega_{0} ||_1  \right).
\end{eqnarray}
Now let us consider the case when $p_{C}(0) = p_{C}(1) = \frac{1}{2}$. In this case, the above expression simplifies to
\begin{eqnarray}
H(C|E)_{\rho_{CE}} = \int_{0}^{1} \frac{\dd s}{2 s\ln 2 } \left( 1 - \frac{1}{2}|| s  \omega_0   - (1 -s)  \omega_1 ||_1  - \frac{1}{2}|| s   \omega_{1} - (1 -s) \omega_{0}||_1 \right). 
\end{eqnarray}
Noting that 
\begin{eqnarray}
 p_{\error}(C|E)_{\bar{\rho}_{CE}(s)}  = \frac{1}{2} (1 - ||s \omega_0 - (1 - s) \omega_{1} ||_1),
\end{eqnarray}
proves the claim for the special case. 
\end{proof}
\section{Upper and lower bounds on the von Neumann entropy using the integral representation}\label{app: upper and lower bounds}
The integral representation above could be used to derive upper and lower bounds on the entropy. In particular, we find the following lower bounds on the integral expression of the entropy. For this section, we drop the subscript $CE$ for convenience from now on when referring to the state $\rho_{CE}$.

\begin{lemma}\label{lemm: known_lower_bounds_using_the_integral_representation}
Let $\rho = \sum_{c} \frac{1}{2} \proj{c} \ot \omega_{c}$, then 
\begin{eqnarray}\label{eqn: fidelity_bound}
  H(C|E)_{\rho} &\geq& \int_{0}^{1} \frac{\dd s}{2s \ln 2} \left(1 - \sqrt{1 - 4 s (1 -s) F( \omega_0  ,  \omega_1 )^2} \right)=  1 -  h_2\left(\frac{1}{2} {+} \frac{F(\omega_0 , \omega_1)}{2} \right)   \\ 
  H(C|E)_{\rho} &\geq&  2 (1 - 2^{-H_{\min}(C|E)_{\rho }}) \geq H_{\min}(C|E)_{\rho}.  
\end{eqnarray} 
\end{lemma}
\begin{proof}
Let us define a family of states parameterised by $s \in [0,1]$:
\[
\bar{\rho}(s) := s \proj{0} \ot \omega_0 + (1 - s) \proj{1} \ot \omega_1.
\]
We recall a well-known lower bound on the optimal error probability for binary state discrimination (see, e.g., \cite[Theorem~9.3.1]{wilde2013quantum}):
\[
p_{\error}(C|E)_{\bar{\rho}(s)} \geq \frac{1}{2} - \frac{1}{2} \sqrt{1 - 4s(1-s) F(\omega_0, \omega_1)^2}.
\]
Substituting this bound into the integral representation of the entropy (see Eq. \eqref{eqn: integral_representation_of_entropy}), we obtain
% 2 of denominator should not be removed here too
\[
H(C|E)_\rho \geq \int_0^1 \frac{\dd s}{2 s \ln 2} \left(1 - \sqrt{1 - 4s(1-s) F(\omega_0, \omega_1)^2} \right).
\]
Using Lemma~\ref{lem:integral_expression_binary_entropy}, which states
\[
h_2\left(\frac{1}{2} + \frac{x}{2} \right) = 1 -\int_0^1 \frac{\dd s}{2 s \ln 2} \left(1 - \sqrt{1 - 4s(1-s)x^2} \right),
\]
we prove the first bound.

For the second bound, let $\Lambda^*_s$ be the measurement that minimises the error probability in distinguishing $\omega_0$ and $\omega_1$ for given $s$. Then we get the following bound
\begin{align}
p_{\error}(C|E)_{\bar{\rho}(s)} &=  s \hspace{0.2em}\tr(\Lambda^*_s \omega_0) + (1 - s) \tr((\id - \Lambda^*_s) \omega_1) \nonumber \\
&\geq  2 \min\{s, 1 - s\} \cdot p_{\error}(C|E)_{\bar{\rho}(1/2)} \label{eqn: error_probability_bound}
\end{align}
since $s , (1 - s) \geq \min\{ s  ,  1- s\}$ and $1/2  \left( \tr(\Lambda^*_s \omega_0 ) +   \tr((\id - \Lambda^*_s) \omega_1) \right) \geq p_{\error}(C|E)_{\bar{\rho}(1/2)} $.
This allows us to bound the integral: 
\[
\int_0^1 \frac{\dd s}{2 s \ln 2} p_{\error}(C|E)_{\bar{\rho}(s)} \geq p_{\error}(C|E)_{ \bar{\rho}(1/2)} \cdot \int_0^1 \frac{\dd s}{s \ln 2} \min\{s, 1 - s\}.
\] 
By symmetry, we also get the same bound for $\int_0^1 \frac{\dd s}{2 s \ln 2} p_{\error}(C|E)_{\bar{\rho}(1 - s)} $. Noting that $\int_{0}^{1} \dd s \min\{ s , 1-s\} s^{-1} = \ln 2$ and that $\rho = \bar{\rho}(1/2)$ gives 
\[
H(C|E)_\rho \geq 2 \cdot p_{\error}(C|E)_\rho.
\]
Now, using the expression $p_{\error}(C|E)_\rho = 1 - 2^{-H_{\min}(C|E)_\rho},$
we obtain
\[
H(C|E)_\rho \geq 2\left(1 - 2^{-H_{\min}(C|E)_\rho} \right).
\]
Finally, since $2(1 - 2^{-x}) \geq x$ for all $x \in [0,1]$, and $H_{\min}(C|E)_\rho \leq \log \dim C = 1$, we conclude
\[
H(C|E)_\rho \geq H_{\min}(C|E)_\rho.
\]
\end{proof}

The first lower bound involving the fidelity coincides with the result from \cite{roga2010universal} and was also employed by Tan et al.\ in the context of device-independent advantage distillation \cite{tan2020advantage}. The second bound, which involves the min-entropy, matches the result in \cite{briet2009properties} and was later used by Stasiuk et al.\ \cite{stasiuk2022quantum} in their work connecting advantage distillation to the quantum Chernoff bound. We remark that both bounds emerge naturally from the integral expression of the entropy in terms of error probabilities, requiring only introductory tools in quantum information theory ~\cite{wilde2013quantum}.  

We further remark that our method can yield a tighter bound than either of the two individual bounds alone. This follows by combining the two lower bounds on the error probability pointwise in $s$, enabling us to derive a stronger lower bound on the conditional von Neumann entropy in cases when both the fidelity and the equal-prior error probability are known. Specifically,
\[
p_{\error}(C|E)_{\bar{\rho}(s)} \geq \max \left\{ \frac{1}{2} - \frac{\sqrt{1 - 4s(1 - s) F(\omega_0, \omega_1)^2}}{2}, \quad 2 \min\{s, 1-s\} \, p_{\error}(C|E)_{\rho} \right\}.
\]
This bound is guaranteed to exceed both known lower bounds on the von Neumann entropy, since taking the maximum of the two ensures it is at least as strong as the better of the individual bounds. We illustrate this point in Fig. 1 in the main text.

We now turn to the problem of upper bounding the key rate, which amounts to upper bounding the conditional entropy \( H(C|E) \). Leveraging the integral representation, we can obtain such bounds by bounding either the error probability or the min-entropy. The following result provides an upper bound on the error probability in terms of the quantum Chernoff bound:

\begin{lemma}
The following upper bound holds:
\begin{eqnarray}
    H(C|E)_{\rho} \leq \frac{\pi}{2 \ln 2}  \left(\frac{Q_{\alpha}(\omega_0 || \omega_1)}{\sin(\pi \alpha)} \right) 
\end{eqnarray}
for every \(\alpha \in (0,1)\).
\end{lemma}

\begin{proof}
We upper bound the integral representation of \( H(C|E) \) by bounding the error probability \( p_{\error}(C|E)_{\bar{\rho}(s)} \) for each \( s \in [0,1] \). This can be achieved using the inequality established in \cite{audenaert2007discriminating}:
\begin{eqnarray*}
    \tr(A^{\alpha} B^{1-\alpha}) \geq \frac{1}{2} \tr(A + B - |A - B|),
\end{eqnarray*}
which holds for all positive semidefinite operators \( A, B \) and \( \alpha \in (0,1) \).  Setting \( A = s \omega_0 \) and \( B = (1 - s) \omega_1 \), we obtain the bound:
\begin{eqnarray*}
    p_{\error}(C|E)_{\bar{\rho}(s)} \leq s^{ \alpha} (1 - s)^{1-\alpha} Q_{\alpha}(\omega_0 \| {\omega}_1),
\end{eqnarray*}
for {any} \( \alpha \in (0, 1) \). Similarly, for the swapped roles of \( \omega_0 \) and \( \omega_1 \), we have:
\begin{eqnarray}
    p_{\error}(C|E)_{\bar{\rho}(1 - s)} \leq (1 - s)^{\beta} s^{1 - \beta} Q_{ \beta}(\omega_0 \| \omega_1)
\end{eqnarray}
for $ \beta \in (0,1)$. Now, substituting the upper bound on \( p_{\error}(C|E)_{\bar{\rho}(s)} \) into the integral representation of entropy gives:
\begin{eqnarray*}
    H(C|E)_{\rho} \leq Q_{\alpha}({\omega}_0 \| {\omega}_1) \left( \int_{0}^{1} \frac{\dd s}{2 s \ln 2}   s^{\alpha} (1 - s)^{1 - \alpha}   \right) + Q_{\beta}({\omega}_0 || {\omega}_1) \left( \int_{0}^{1} \frac{\dd s}{2 s \ln 2} s^{1- \beta} (1 - s)^{\beta}\right)  .
\end{eqnarray*}
The integrals above can be evaluated using the identity:
\begin{eqnarray*}
    \int_0^1 \frac{\dd s}{s \ln 2} s^{\alpha} (1 - s)^{1 - \alpha} = \frac{\pi}{\sin(\pi \alpha)} \frac{1 - \alpha}{\ln 2},
\end{eqnarray*}
which gives  
\begin{eqnarray*}
    H(C|E)_{\rho} \leq \frac{\pi}{2 \ln 2}\left(  Q_{\alpha}({\omega}_0 \| {\omega}_1) \frac{1 - \alpha}{\sin(\pi \alpha)} + Q_{\beta}({\omega}_0 || {\omega}_1) \frac{\beta}{\sin(\pi \beta)} \right) .
\end{eqnarray*}
Let $\beta=\alpha \in (0,1)$ to get
\begin{eqnarray*}
     H(C|E)_{\rho} \leq \frac{\pi}{2\ln 2} \frac{Q_{\alpha}({\omega}_0 || {\omega}_1)}{\sin(\pi  \alpha)},
\end{eqnarray*}
which is desired. 
\end{proof}
\section{Application of the integral representation to the AD-QKD protocol}\label{app: integral representation QKD}

We now apply the integral representation of
Theorem~\ref{thm: CQ_integral_representation} to the QKD protocol with an advantage
distillation step, which we refer to as the AD-QKD protocol. All classical and quantum
registers are as defined in the main text. Throughout, $n$ denotes the blocklength,
$\epsilon \in (0, \tfrac{1}{2})$ is a parameter identified from the statistics, $D$ the register recording whether the
block is accepted, and
\begin{equation}
    \delta_{n} := \frac{\epsilon^{n}}{(1-\epsilon)^{n} + \epsilon^{n}} .
\end{equation}

As discussed in the main text, the condition to derive a positive key rate from the protocol is given by
\cite{tan2020advantage, bae2007key}:
\begin{equation}\label{eqn: DW_condition}
    H(C|\B{ETM}; D = 1) - H(C|C'; D = 1) > 0 .
\end{equation}
Conditioned on acceptance, Bob's guess $C'$ differs from $C$ with probability exactly
$\delta_{n}$, so that $H(C|C'; D = 1) = h_{2}(\delta_{n})$. Taking all states to be
conditioned on acceptance from now on, and suppressing $D = 1$ from the notation,
Eq.\eqref{eqn: DW_condition} is therefore equivalent to the ratio condition
\begin{equation}\label{eqn: R_n}
    \mathscr{R}_{n} := \frac{H(C|\B{ETM})_{\rho}}{h_{2}(\delta_{n})} > 1 ,
\end{equation}
where
\begin{equation}\label{eqn: rho_CETM}
    \rho = \sum_{\B{m}} p(\B{m}) \proj{\B{m}} \ot \rho_{C\B{ET}, \B{M} = \B{m}} ,
    \qquad
    \rho_{C\B{ET}, \B{M} = \B{m}}
    = \sum_{c \in \{0,1\}} \tfrac{1}{2} \proj{c} \ot \omega_{c}^{[\B{m}]} ,
\end{equation}
\begin{align}
    \omega_{0}^{[\B{m}]}
    &= (1 - \delta_{n}) \rho_{\B{ET}|\B{m}\B{m}}
     + \delta_{n} \rho_{\B{ET}|\B{m}\bar{\B{m}}} , \label{eqn: omega_0_m} \\
    \omega_{1}^{[\B{m}]}
    &= (1 - \delta_{n}) \rho_{\B{ET}|\bar{\B{m}}\bar{\B{m}}}
     + \delta_{n} \rho_{\B{ET}|\bar{\B{m}}\B{m}} , \label{eqn: omega_1_m}
\end{align}
and $p(\B{m})$ is the probability that the message string $\B{m}$ is announced. The key
rate of the protocol, excluding the parameter estimation rounds, is \cite{Renner}
\begin{equation}\label{eqn: key_rate}
    \left( H(C|\B{ETM})_{\rho} - h_{2}(\delta_{n}) \right)
    \cdot \frac{(1-\epsilon)^{n} + \epsilon^{n}}{n} ,
\end{equation}
so that whether a key can be distilled at blocklength $n$ is governed entirely by
$\mathscr{R}_{n}$.
The goal of this section is to rewrite $\mathscr{R}_{n}$ as an integral over binary
hypothesis testing problems. One obstacle is that $H(C|\B{ETM})_{\rho}$ appears to depend
on the announced string $\B{m}$ through the states $\omega_{0}^{[\B{m}]}$ and
$\omega_{1}^{[\B{m}]}$. Such a dependence would also complicate the security analysis,
since each message string would have to be treated separately. Following
\cite{stasiuk2022quantum}, we show that it suffices to analyse the single string
$\B{m} = (0, 0, \dots, 0)$.

\begin{lemma}\label{corr: neccessary_and_sufficiant_condition}
Let $\mathscr{R}_{n}$ be as in Eq. \eqref{eqn: R_n}. Then
\begin{equation}\label{eqn: Rn_integral}
    \mathscr{R}_{n} = \frac{1}{h_{2}(\delta_{n})} \int_{0}^{1} \frac{\dd s}{2 s\ln 2}
    \left[ p_{\error}(C|\B{ET})_{\bar{\rho}(s)}
         + p_{\error}(C|\B{ET})_{\bar{\rho}(1-s)} \right] ,
\end{equation}
where $\bar{\rho}(s) := s \proj{0}_{C} \ot \omega_{0}^{(n)}
+ (1-s) \proj{1}_{C} \ot \omega_{1}^{(n)}$ and
\begin{align}
    \omega_{0}^{(n)} &= (1 - \delta_{n}) \rho_{ET|00}^{\ot n}
                      + \delta_{n} \rho_{ET|01}^{\ot n} , \\
    \omega_{1}^{(n)} &= (1 - \delta_{n}) \rho_{ET|11}^{\ot n}
                      + \delta_{n} \rho_{ET|10}^{\ot n} .
\end{align}
Here $\rho_{ET|ab}$ is the single-round post-measurement state held by the adversary, as
defined in the main text.
\end{lemma}

\begin{proof}
Since the register $\B{M}$ is classical,
\begin{equation}
    H(C|\B{ETM})_{\rho}
    = \sum_{\B{m}} p(\B{m}) \, H(C|\B{ET})_{\rho_{C\B{ET}, \B{M} = \B{m}}} ,
\end{equation}
and each conditional state $\rho_{C\B{ET}, \B{M} = \B{m}}$ has equal priors on $C$ by Eq. \eqref{eqn: rho_CETM}. Applying Theorem~\ref{thm: CQ_integral_representation} to each term
therefore gives
\begin{equation}\label{eqn: Rn_m_dependent}
    \mathscr{R}_{n} = \sum_{\B{m}} \frac{p(\B{m})}{h_{2}(\delta_{n})}
    \int_{0}^{1} \frac{\dd s}{2 s\ln 2}
    \left[ p_{\error}(C|\B{ET})_{\tilde{\rho}_{\B{m}}(s)}
         + p_{\error}(C|\B{ET})_{\tilde{\rho}_{\B{m}}(1-s)} \right] ,
\end{equation}
where
\begin{equation}
    \tilde{\rho}_{\B{m}}(s) := s \proj{0}_{C} \ot \omega_{0}^{[\B{m}]}
    + (1-s) \proj{1}_{C} \ot \omega_{1}^{[\B{m}]} .
\end{equation}
Since the weights $p(\B{m})$ sum to one, it suffices to show that
$p_{\error}(C|\B{ET})_{\tilde{\rho}_{\B{m}}(s)} = p_{\error}(C|\B{ET})_{\bar{\rho}(s)}$ for
every $\B{m}$ and every $s \in [0,1]$.

Following \cite{stasiuk2022quantum}, introduce the unitary
\begin{equation}
    U^{(\B{m})} = \bigotimes_{i=1}^{n} (\id_{E} \ot X)^{m_{i}} ,
\end{equation}
where $X$ is the Pauli-$X$ operator acting on the register $T_{i}$. The symmetrisation step
of the protocol guarantees that flipping $T_{i}$ exchanges the two branches of the
single-round conditional state, so that
\begin{equation}\label{eqn: flip_relation}
    (\id_{E} \ot X) \, \rho_{ET|ab} \, (\id_{E} \ot X) = \rho_{ET|\bar{a}\bar{b}}
    \qquad \text{for all } a, b \in \{0,1\} .
\end{equation}
Since $\rho_{\B{ET}|\B{m}\B{m}} = \bigotimes_{i=1}^{n} \rho_{ET|m_{i}m_{i}}$, and similarly
for the other three strings, applying  Eq. \eqref{eqn: flip_relation} factor by factor yields
\begin{align*}
    U^{(\B{m})} \rho_{\B{ET}|\B{m}\B{m}} (U^{(\B{m})})^{\dagger}
        &= \rho^{\ot n}_{ET|00} , &
    U^{(\B{m})} \rho_{\B{ET}|\bar{\B{m}}\bar{\B{m}}} (U^{(\B{m})})^{\dagger}
        &= \rho^{\ot n}_{ET|11} , \\
    U^{(\B{m})} \rho_{\B{ET}|\B{m}\bar{\B{m}}} (U^{(\B{m})})^{\dagger}
        &= \rho^{\ot n}_{ET|01} , &
    U^{(\B{m})} \rho_{\B{ET}|\bar{\B{m}}\B{m}} (U^{(\B{m})})^{\dagger}
        &= \rho^{\ot n}_{ET|10} ,
\end{align*}
and hence, by linearity,
$U^{(\B{m})} \omega_{c}^{[\B{m}]} (U^{(\B{m})})^{\dagger} = \omega_{c}^{(n)}$ for
$c \in \{0,1\}$.

Finally, using the fact that trace norm is invariant under unitary conjugation we get
\begin{equation*}
    \big\| s \, \omega_{0}^{[\B{m}]} - (1-s) \, \omega_{1}^{[\B{m}]} \big\|_{1}
    = \big\| s \, \omega_{0}^{(n)} - (1-s) \, \omega_{1}^{(n)} \big\|_{1}.
\end{equation*}
Therefore $p_{\error}(C|\B{ET})_{\tilde{\rho}_{\B{m}}(s)}
= p_{\error}(C|\B{ET})_{\bar{\rho}(s)}$ for every $\B{m}$, and likewise at $1-s$.
Substituting into Eq. \eqref{eqn: Rn_m_dependent} gives Eq. \eqref{eqn: Rn_integral}. Note that
$\bar{\rho}(s) = \tilde{\rho}_{\B{m}}(s)$ for the particular choice
$\B{m} = (0, 0, \dots, 0)$, so the lemma may be read as saying that the security analysis
may be carried out for that string alone.
\end{proof}

\begin{notation}
From now on we write $E$ for the totality of the quantum side information available to the
adversary. In particular we write $E$ in place of $\B{ET}$, $ET$ and $\B{ETM}$ wherever no
ambiguity arises.
\end{notation}
\section{Condition to derive a key in the asymptotic limit}\label{app: asymptotic n_and_s}
We are now in a position to prove one of the main results of this work: a necessary and sufficient condition for obtaining a positive key rate in the asymptotic limit of the advantage distillation protocol. We emphasise that the asymptotic positivity of the key rate refers to the ratio $\mathscr{R}_{n}$ being strictly greater than one.  Note, however, that the key rate itself tends to zero as \( n \rightarrow \infty \).

We begin by deriving the result for the general case where the shared states could be mixed. Subsequently, we extend the analysis to the case where Alice and Bob (parties performing the QKD protocol) share a qubit pair and perform projective measurements. In this special case, we obtain a significantly stronger result than in the general mixed-state scenario.

\subsection{General result}\label{app: asymptotic n_and_s_general}
\begin{lemma}\label{lemm: general condition}
Suppose that the post-measurement states $\rho_{E|00}$ and $\rho_{E|11}$ satisfy 
\begin{equation}
  Q (\rho_{E|00} , \rho_{E|11}) \leq  \frac{\epsilon}{1 - \epsilon}.
\end{equation}
Then 
\begin{eqnarray}
 \lim_{n \rightarrow \infty} \mathscr{R}_{n} \leq 1.
\end{eqnarray}
\end{lemma}

\begin{proof}
We begin with the observation that for any \( s \in [0,1] \),
\begin{equation*}
p_{\error}(C|E)_{\bar{\rho}(s)} \leq s,
\end{equation*}
where $\bar{\rho}(s)=s\ketbra{0}{0}\otimes\omega_0^{(n)}+(1-s)\ketbra{1}{1}\otimes\omega_1^{(n)}$. This holds because
\begin{align}
p_{\error}(C|E)_{\bar{\rho}(s)} &= \inf_{\id \succeq \Lambda \succeq 0} \left( s \, \tr [\Lambda \omega_0^{(n)}] + (1-s) \, \tr [(\id - \Lambda) \omega_1^{(n)}] \right).
\end{align}
By choosing \(\Lambda = \id \), we obtain the upper bound \(s\). Next, observe that 
\begin{align}
p_{\error}(C|E)_{\bar{\rho}(s)} 
= \frac{1}{2}\left( 1 - \left\| (1 - \delta_n) \Delta_0(s) + \delta_n \Delta_1(s) \right\|_1 \right) \nonumber,
\end{align}
where
\begin{align}
\Delta_0(s) &= s \, \rho_{E|00}^{\otimes n} - (1-s) \, \rho_{E|11}^{\otimes n},  \nonumber\\
\Delta_1(s) &= s \, \rho_{E|01}^{\otimes n} - (1-s) \, \rho_{E|10}^{\otimes n} \nonumber.
\end{align}
We therefore obtain the bound
\begin{align}
p_{\error}(C|E)_{\bar{\rho}(s)} 
&\leq \frac{1}{2}\left( 1 - \| \Delta_0(s) \|_1 + \delta_n \, \| \Delta_0(s) - \Delta_1(s) \|_1 \right).
\end{align}
Using the triangle inequality and the fact that \(\| \Delta_0(s) \|_1, \| \Delta_1(s) \|_1 \leq 1\), we have
\[
\| \Delta_0(s) - \Delta_1(s) \|_1 \leq \| \Delta_0(s) \|_1 + \| \Delta_1(s) \|_1 \leq 2.
\]
Thus,
\begin{align}
p_{\error}(C|E)_{\bar{\rho}(s)} 
\leq \frac{1}{2}\left( 1 - \| \Delta_0(s) \|_1 + 2 \delta_n \right) \nonumber.
\end{align}
Now we apply the bound by \cite{audenaert2007discriminating} to get that for any \(\eta \in (0,1)\),
\begin{align}
\| \Delta_0(s) \|_1 \geq 1 - 2 s^{\eta} (1-s)^{1-\eta} Q_{\eta}(\rho_{E|00} \,\|\, \rho_{E|11})^n. \nonumber
\end{align}
Thus,
\begin{align}
p_{\error}(C|E)_{\bar{\rho}(s)} 
&\leq  s^{\eta} (1-s)^{1-\eta} Q_{\eta}(\rho_{E|00} \,\|\, \rho_{E|11})^n + \delta_n. \nonumber
\end{align}
Following Lemma \ref{corr: neccessary_and_sufficiant_condition}, we need to show that
\begin{eqnarray}
 \lim_{n \rightarrow \infty}  \frac{1}{h_{2}(\delta_{n})} \int_{0}^{1} \frac{\dd s}{2s\ln 2 } \left[ p_{\error}(C|E)_{\bar{\rho}(s)}  + p_{\error}(C|E)_{\bar{\rho}(1 -s)} \right] \leq  1. 
\end{eqnarray}
We start by considering the integral 
\begin{align}
I_n := \int_0^1 \frac{\mathrm{d}s}{2s \ln 2} \, p_{\error}(C|E)_{\bar{\rho}(s)}. \nonumber
\end{align}
Splitting the integral as
\begin{align*}
I_n &= \int_0^{\delta_n} \frac{\mathrm{d}s}{2 s \ln 2} \, p_{\error}(C|E)_{\bar{\rho}(s)} + \int_{\delta_n}^1 \frac{\mathrm{d}s}{2 s \ln 2} \, p_{\error}(C|E)_{\bar{\rho}(s)},
\end{align*}
we obtain
\begin{align*}
I_n \leq \frac{\delta_n}{2 \ln 2} - \frac{\delta_n \log_2 \delta_n}{2} 
+ Q_{\eta}(\rho_{E|00} \,\|\, \rho_{E|11})^n \int_{\delta_n}^1 \frac{\mathrm{d}s}{2 \ln 2} \, s^{\eta-1} (1-s)^{1-\eta}.
\end{align*}
Thus, we find that
\begin{align}
\lim_{n \to \infty} \frac{I_n}{h_2(\delta_n)} 
&\leq \lim_{n \to \infty} \left( \frac{Q_{\eta}(\rho_{E|00} \,\|\, \rho_{E|11})}{\frac{\epsilon}{1-\epsilon}} \right)^n \cdot  {\frac{\left(\frac{\epsilon}{1-\epsilon}\right)^n}{h_2(\delta_n)} \cdot} \left( \int_{\delta_n}^1 \frac{\mathrm{d}s}{2 \ln 2} \, s^{\eta-1} (1-s)^{1-\eta} \right) \nonumber \\
&\quad + \lim_{n \to \infty} \frac{\delta_n}{{2 \ln 2}h_2(\delta_n)} + \lim_{n \to \infty} \frac{- \delta_n \log_2 \delta_n}{2 h_2(\delta_n)}. \nonumber
\end{align}
Notice that
\begin{align*}
&\lim_{n \to \infty} \frac{\delta_n}{h_2(\delta_n)} = 0, \qquad 
\lim_{n \to \infty} \frac{- \delta_n \log_2 \delta_n}{h_2(\delta_n)} = 1. 
\end{align*}
This means that to compute the limit $\lim_{n \rightarrow \infty} I_n/ h_2(\delta_n )$, it suffices to compute
\begin{align*}
\lim_{n \to \infty} \left( \frac{Q_{\eta}(\rho_{E|00} \,\|\, \rho_{E|11})}{\frac{\epsilon}{1-\epsilon}} \right)^n \cdot  {\frac{\left(\frac{\epsilon}{1-\epsilon}\right)^n}{h_2(\delta_n)} \cdot} \left( \int_{\delta_n}^1 \frac{\mathrm{d}s}{2 \ln 2} \, s^{\eta-1} (1-s)^{1-\eta} \right).  
\end{align*}
Recall that from Lemma \ref{lemm: ratio} we have that
\begin{align*}
\lim_{n \rightarrow \infty} \frac{\left( \frac{\epsilon }{1 - \epsilon}\right)^n}{h_2(\delta_{n})} = \lim_{n \rightarrow \infty} \frac{(\frac{\epsilon}{1 - \epsilon})^{n}}{- n \left( \frac{\epsilon}{1 - \epsilon} \right) \log\left( \frac{\epsilon}{1 - \epsilon}\right) } = 0.  
\end{align*}
Thus if the sequence $ \left( Q_{\eta}(\rho_{E|00} || \rho_{E|11}) / (\frac{\epsilon}{1 - \epsilon}) \right)^{n} $ is bounded,  then 
\[
\lim_{n \to \infty} \frac{I_n}{h_2(\delta_n)} \leq \frac{1}{2}.
\]
The sequence $\left( Q_{\eta}(\rho_{E|00} || \rho_{E|11})/(\frac{\epsilon}{1  - \epsilon}) \right)^n $ is bounded whenever 
\begin{eqnarray}
    \inf_{\eta \in (0 ,1)} Q_{\eta}(\rho_{E|00}  || \rho_{E|11}) \leq \frac{\epsilon}{1 - \epsilon}.
\end{eqnarray}
A similar argument holds for the integral with \(p_{\error}(C|E)_{\bar{\rho}(1-s)}\), yielding
\[
\lim_{n \to \infty} \frac{1}{h_2(\delta_n)} \int_0^1 \frac{\mathrm{d}s}{2 s \ln 2} \left( p_{\error}(C|E)_{\bar{\rho}(s)} + p_{\error}(C|E)_{\bar{\rho}(1-s)} \right) \leq 1.
\]
This shows that $\lim_{n \rightarrow \infty} \mathscr{R}_{n} \leq 1$ when $  Q(\rho_{E|00} , \rho_{E|11}) \leq \frac{\epsilon}{1 - \epsilon}.$
\end{proof}
\subsection{Results for qubit pairs}\label{app: asymptotic n_and_s_qubit}

We now turn to the case where Alice and Bob share a qubit pair and perform projective measurements. Our approach is to first develop the necessary tools and bounds in the most general setting, without assuming any special structure. Once such bounds are established, we adapt them to the qubit case.

\subsubsection{Bounding the error using multiple hypothesis testing} 

In order to obtain an upper bound on the error probability for the
binary ensemble
\[
\left\{
s\omega_0^{(n)},\,
(1-s)\omega_1^{(n)}
\right\},
\]
we fine-grain the first hypothesis by decomposing
\[
\omega_0^{(n)}
=
(1-\delta_n)\rho_{E|00}^{\otimes n}
+
\delta_n\rho_{E|01}^{\otimes n}.
\]
We therefore consider the three-state ensemble
\[
\left\{
s(1-\delta_n)\rho_{E|00}^{\otimes n},\,
s\delta_n\rho_{E|01}^{\otimes n},\,
(1-s)\omega_1^{(n)}
\right\}.
\]
By Lemma~\ref{lemm: min-entropy_monotonicity}, fine-graining an ensemble in this
way cannot decrease its minimum error probability. Hence,
\[
p_{\error}(C|E)_{\bar{\rho}(s)}
\leq
p_{\error}(C|E)_{\rho'(s)}.
\]
An upper bound on the error probability of the fine-grained ensemble
therefore also provides an upper bound on the error probability of the
original binary ensemble. As we will see in later sections, this approach yields more tractable upper bounds in terms of Petz--Rényi quantities, 
which are essential for proving the upper bounds on the error probability in the qubit case.

We therefore begin by considering the conditional error probability of the following cq state: 
\begin{equation}\label{eqn: three_state_discrimination}
\rho'(s) :=  s(1-\delta_n)\proj{0} \otimes\,\rho_{E|00}^{\otimes n}  
+ s\delta_n\proj{1} \otimes \,\rho_{E|01}^{\otimes n}  
+ (1-s)\proj{2} \otimes \omega_1^{(n)}.
\end{equation}
We now state a theorem from \cite{li2016discriminating} that allows us to upper bound the error probability for the multiple-hypothesis-testing problem: 
\begin{theorem}\label{thm: Li}(Li \cite{li2016discriminating})\label{thm: discriminating_multiple_hypothesis}
Let $ A_1 ,  A_{2} , \cdots , A_{r} \in \M{S}(\M{H})$. For all $1 \leq i \leq r$, let $A_{i } = \sum_{k = 1}^{T_{i}} \lambda_{ik} \Pi_{ik}$ be the spectral decomposition of $A_{i}$ and write $T = \max \{ T_1 , T_{2} , \cdots , T_{r} \}$. {Let $\rho = \sum_{i} p_{i} \proj{i}_{i} \ot A_{i}$ be the cq state corresponding to the ensemble $\{ p_{1} A_{1} , p_2 A_{2} ,  \cdots , p_{r} A_{r}  \}$}. Then,

\begin{eqnarray}
p_{\error}( {C|E})_{{\rho}} &\leq&  10 (r - 1)^2 T^2 \sum_{(i,j): i < j} \sum_{k,l} \min \{ p_i\lambda_{ik} , p_j\lambda_{jl} \} \tr \left( {\Pi}_{ik} {\Pi}_{jl} \right)   \\ 
&\leq& 10 (r - 1)^2 T^2 \sum_{(i,j): i < j} p_{i}^{\alpha_{ij}} p_{j}^{1 - \alpha_{ij}}  Q_{\alpha_{ij}}({A_i} ||{A_j} )
\end{eqnarray}
where $\alpha_{{ij}} \in {[}0 ,1{]}$ {and $p_i \geq 0$ with $\sum_i p_i=1$}.
\end{theorem}

Using the theorem above, we prove our main result:

\begin{lemma}\label{thm: main_result_general}
{For $\bar{\rho}(s)=s\proj{0} \ot \omega_0^{(n)}+(1-s)\proj{1} \ot \omega_1^{(n)}$, $p_{\error}(C|E)_{\bar{\rho}(s)}$ can be upper bounded as} 
\begin{eqnarray}
 p_{\error}({C|E})_{{\bar{\rho}(s)}} \leq p(n) \left( s^{1 - \alpha} (1 - s)^{\alpha} A_{\alpha} + s^{1- \beta} (1 - s)^{\beta} B_{\beta} +  s \delta_{n} \right) 
\end{eqnarray}
where $p(n)$ is some function, the variables $\alpha , \beta $ are in the interval $(\frac{1}{2} ,1)$ and 
\begin{eqnarray}
A_{\alpha} &=& (1 - \delta_{n}^2)^{1 - \alpha} \tilde{Q}_{\alpha}(\rho_{E|11}|| \rho_{E|00})^{n} + \delta_{n}^{{\alpha}} (1 + \delta_{n})^{1 - \alpha}   \tilde{Q}_{\alpha}(\rho_{E|{10}} || \rho_{E|{00}})^{n} \\ 
B_{\beta} &=& \delta_{n}^{1 - \beta} (1 + \delta_{n})^{1 - \beta} \tilde{Q}_{\beta}(\rho_{E|11} || \rho_{E|01})^{n} + \delta_{n}  (1 + \delta_{n})^{1 - \beta} \tilde{Q}_{\beta}(\rho_{E|{10}} || \rho_{E|{01}})^{n}.      
\end{eqnarray}
\end{lemma}

\begin{proof}
The proof begins by establishing an upper bound on the error probability \( p_{\error}({C|E})_{{\bar{\rho}(s)}} \) in terms of the error probability \( p_{\error}({C|E})_{{\rho'(s)}} \), where the {cq state} \( {\rho'(s)} \) is defined as in Eq. \eqref{eqn: three_state_discrimination}. {This step follows from Lemma \ref{lemm: min-entropy_monotonicity}, which intuitively states that the error probability for a given ensemble cannot decrease under "fine-graining" the ensemble}. 

Next, we apply Theorem \ref{thm: discriminating_multiple_hypothesis} to obtain an upper bound on \( p_{\error}({C|E})_{{\rho'(s)}} \) in terms of the Petz--Rényi quantities \( Q_{\alpha}(\cdot \Vert \cdot) \). The function \( f(r, T){:=10(r-1)^2T^2} \) appearing in Theorem \ref{thm: discriminating_multiple_hypothesis} depends on \( T \), the maximum number of distinct eigenvalues among the states in the ensemble. The type counting lemma bounds (see Li \cite{li2016discriminating}) the number of distinct eigenvalues of a $d$-dimensional state raised to the $n$-th tensor power by
\[
\Omega(A^{\otimes n}) \leq (n+1)^d
\]
where $\Omega$ denotes the number of distinct eigenvalues. Since \( \rho_{E|00}^{\otimes n} \) and \( \rho_{E|01}^{\otimes n} \) are each such tensor powers, this bounds their number of distinct eigenvalues directly. The state \( \omega_1^{(n)} \), being a mixture of two tensor powers rather than a tensor power itself, need not satisfy so tight a bound in general; it suffices for our purposes here to note that, since it is supported on a finite-dimensional Hilbert space, its number of distinct eigenvalues is trivially bounded by the dimension of that space, which is itself some finite function of \( n \). Consequently, \( T \) is bounded by some finite function of \( n \), and hence so is \( f(r,T) = 10(r-1)^2T^2 \) with \( r = 3 \). Thus, for some function \( p(n) \),
\begin{eqnarray}
p_{\error}({C|E})_{{\bar{\rho}(s)}} &\leq& p(n) \Big( A_{1} + A_{2} + A_{3}\Big) \label{eqn: multiple_hypothesis}, 
\end{eqnarray}
where 
\begin{eqnarray}
    A_{1} &:=& (1 - s)^{\alpha} s^{1 - \alpha} {(1-\delta_{n})}^{1- \alpha} Q_{\alpha}(\omega_{1}^{(n)} || \rho_{E|00}^{\ot n} ) \\ 
    A_{2}  &:=&  (1 - s)^{\beta} s^{1 - \beta} \delta_{n}^{1 - \beta}  Q_{\beta}(\omega_{1}^{(n)} || \rho_{E|01}^{\ot n} )  \\ 
    A_{3}  &:=&  s \delta_{n}^{\gamma} (1 - \delta_{n})^{1 - \gamma} Q_{\gamma}(\rho_{E|01}^{\ot n} || \rho_{E|00}^{\ot n} )  
\end{eqnarray}
Note that Eq. \eqref{eqn: multiple_hypothesis} holds for any  $\alpha , \beta, \gamma \in [0 ,1]$. We now upper bound terms in Eq. \eqref{eqn: multiple_hypothesis} individually. We consider bounding $A_{ {3}}$:
\begin{eqnarray}
A_{3}    = s (1 - \delta_{n})^{1 - \gamma} \delta_{n}^{\gamma} Q_{\gamma}(\rho_{E|01} || \rho_{E|00} )^{n}. 
\end{eqnarray}
 Let $\gamma=1$ to get
 \begin{eqnarray}
A_3=s\delta_nQ_1(\rho_{E|01}\|\rho_{E|00})^n\leq s\delta_n,
 \end{eqnarray}
 where we have used $Q_1(\rho_{E|01}\|\rho_{E|00})\leq 1$. For the first and the second term, we first upper bound $Q_{\lambda}(. || . )$ by $\tilde{Q}_{\lambda}(.||.) $ and then use Lemma \ref{lemm: continuity_of_tilde_Q}. For the second term we get 
\begin{eqnarray*}
A_{2} &:=&  s^{1 - \beta} (1 - s)^{ \beta} \delta_{n}^{1 - \beta} Q_{\beta}(\omega_{1}^{(n)} || \rho_{E|01}^{\ot n} )\\
 &\leq& s^{1 - \beta} (1 - s)^{ \beta} \delta_{n}^{1 - \beta} \tilde{Q}_{\beta}(\omega_{1}^{(n)} || \rho_{E|01}^{\ot n} )  \\ 
&\leq&  s^{1 - \beta} (1 -s )^{\beta} \delta_{n}^{1- \beta}\left( \frac{\tilde{Q}_{\beta}(\rho_{E|11} ||\rho_{E|01} )^{n} + \delta^{\beta}_{n}\tilde{Q }_{\beta}(\rho_{E|10} || \rho_{E|01}    )^{n} }{(1 + \delta_{n})^{\beta - 1}} - \delta_{n} \tilde{Q}_{\beta}(\rho_{E|11} || \rho_{E|01} )^{n} \right) \\ 
&{\leq}& s^{1 - \beta} (1 - s)^{\beta} \left(  \frac{\delta_{n}^{1 - \beta}}{(1 + \delta_{n})^{\beta - 1}} \tilde{Q}_{\beta}(\rho_{E|11} || \rho_{E|01})^{n} + \frac{\delta_{n}}{(1 + \delta_{n})^{\beta - 1}} \tilde{Q}_{\beta}(\rho_{E|10}|| \rho_{E|01})^{{n}} \right) 
\end{eqnarray*} 
Now, for the first term we get 
\begin{eqnarray}
A_{1}   \leq (1 - s)^{\alpha} s^{1 - \alpha}  \left( \frac{(1 - \delta_{n})^{1 - \alpha}}{(1 + \delta_{n})^{\alpha - 1}} \tilde{Q}_{\alpha}(\rho_{E| {11}} || \rho_{E|00})^{n} +  \frac{\delta_{n}^{\alpha} {(1 - \delta_n)^{1-\alpha}}}{(1 + \delta_{n})^{\alpha - 1} } \tilde{Q}_{\alpha}(\rho_{E|{10}} || \rho_{E|00})^{n} \right)
\end{eqnarray}
\end{proof} 

The above result simplifies greatly when we restrict the analysis to pure states and projective measurements -- i.e. Alice and Bob share a state $\rho_{Q_A Q_B}$ and their devices perform a projective measurement. For such a case, it can be easily verified that the resulting post-measurement states held by Eve $\rho_{E|ij}$ are also pure provided Eve holds a purification (which can be assumed without any loss of generality).

\begin{theorem}\label{thm:main_result_pure}
Suppose that the conditional states $\rho_{E|ij} = \ket{\psi_{E|ij}}\bra{\psi_{E|ij}}$ are pure, and $F(\rho_{E|ij} , \rho_{E|kl
}) < 1$ for any pair $(i,j) , (k,l)$ such that $(i,j) \neq (k,l)$. Then there exists $\alpha^* \in (\frac{1}{2}, 1) $ and $\beta^* \in (0,1) $ such that $\forall \alpha \in [\alpha^* , 1 )$ and $\beta \in (0,\beta^*]$: 
\begin{eqnarray}
p_{\error}({C|E})_{{\rho'(s)}} \leq  160 s^{1- \alpha} (1 - s)^{\alpha}  F(\rho_{E|00} , \rho_{E|11})^{2n \alpha}   + c_{\alpha,\beta}(s) \left(\frac{\epsilon}{1 - \epsilon} \right)^{n} 
\end{eqnarray}
where $c_{\alpha,\beta}(s)$ is independent of the blocklength $n$ and $c_{\alpha,\beta}(s)/s$ is integrable over the interval $(0 ,1) $.
\end{theorem}

\begin{proof}
Since the states $\rho_{E|ij}$ are pure, we have the following:
\begin{equation}
 \forall \alpha \in (0 ,1): \quad    \tilde{Q}_{\alpha}(\rho_{E|ij} || \rho_{E|kl}) = F(\rho_{E|ij}, \rho_{E|kl})^{2\alpha} , \quad Q_{\alpha}(\rho_{E|ij} || \rho_{E|kl} ) = F(\rho_{E|ij} , \rho_{E|kl})^{2}
\end{equation}
Furthermore, in Lemma ~\ref{thm: main_result_general}, $T$ was bounded only by some finite, and in general dimension-dependent, function of $n$. In the present case, however, since each $\rho_{E|ij}$ is pure, $\rho_{E|ij}^{\otimes n}$ remains pure for every $n$. As a result the state $\omega_{i}^{(n)}$, being a mixture of two pure states is supported on a two dimensional subspace. Hence $T \leq 2$ can be taken regardless of the value of $n$, and $f(r, T) = 10(r-1)^2T^2$, with $r=3$ and $T \leq 2$, is bounded by the constant $160$, rather than growing with $n$ as in Lemma~\ref{thm: main_result_general}. This allows us to derive a significantly tighter bound on the error probabilities for this special case.

Our next observation is that:
\begin{equation*}
    \delta_n :=   \frac{\epsilon^{n}}{(1 - \epsilon)^{n} + \epsilon^{n}}   \leq  \left( \frac{\epsilon}{1 - \epsilon}\right)^{n}.
\end{equation*}
From this point, the proof proceeds in the same manner as in Lemma~\ref{thm: main_result_general}. Setting $\gamma = 1$ in the third term ($A_{3}$) of Eq.~\eqref{eqn: multiple_hypothesis} yields the upper bound $\delta_{n} s$. We can upper bound $\delta_{n}$ by $\left(\frac{\epsilon}{1 - \epsilon}\right)^{n}$ to get the bound 
\begin{eqnarray*}
    A_{3} \leq s \left( \frac{\epsilon}{1 - \epsilon} \right)^n . 
\end{eqnarray*}
It remains to bound $A_1$ and $A_2$. Following \cite{hahn2024bounds} (see Section~6), we obtain  
\begin{equation*}
    \big| Q_{\beta}(\omega_{1}^{(n)} || \rho_{E|01}^{\ot n }) - Q_{\beta}(\rho_{E|11}^{\ot n} || \rho_{E|01}^{\ot n}) \big| 
    \leq  \|\omega_{1}^{(n)} - \rho_{E|11}^{\ot n} \|_1^{\beta} 
    \leq {(2\delta_n)^{\beta}}.
\end{equation*}
Thus, we have  
\begin{align*}
    A_{2} 
    &\leq  \delta_n^{1 - \beta} s^{1- \beta} (1 - s)^{\beta} \left(Q_{\beta}(\rho_{E|11} || \rho_{E|01})^n + {(2\delta_n)}^{\beta} \right)  \\ 
    &= \delta_n^{1 - \beta} s^{1 - \beta} (1 - s)^{\beta} F(\rho_{E|11}, \rho_{E|01})^{2n} + {2^\beta} s^{1 - \beta} (1 - s)^{\beta} \delta_n \\ 
    &\leq  \delta_n^{1 - \beta} s^{1 - \beta} (1 - s)^{\beta} F(\rho_{E|11}, \rho_{E|01})^{2n} + {2^\beta} s^{1 - \beta} (1 - s)^{\beta} \left(\frac{\epsilon}{1 - \epsilon} \right)^n.
\end{align*}
Now note that if $F(\rho_{E|11} , \rho_{E|01})^{2} < 1$, by continuity, there exists $\beta^{{*}} \in (0,1) $ such that
$$ \forall \beta \in (0 , \beta^*]: \left(\frac{\epsilon}{1 - \epsilon}\right)^{\beta} \geq F(\rho_{E|11} , \rho_{E|01})^{2}. $$ Noting that $\delta_n \leq \left(\frac{\epsilon}{1-\epsilon}\right)^n$, we choose any such $\beta$ to get the bound for all $\beta \in ( 0 , \beta^*]$: 
\begin{eqnarray*}
A_{2} \leq  {(1+2^\beta)} s^{{1-\beta}} (1 -s)^{\beta} \left(\frac{\epsilon}{1 - \epsilon} \right)^{n}. 
\end{eqnarray*}
We remark that $A_{2}$ is not bounded using the same technique as it is bounded in Lemma \ref{thm: main_result_general}. This is because we want the bound to hold for $\beta < 1/2$. However, the continuity bounds developed for sandwiched R\'enyi entropy in Lemma \ref{lemm: continuity_of_tilde_Q} do not hold.

We now bound $A_{1}$  from Lemma \ref{thm: main_result_general}
\begin{align*}
    A_{1} 
    &\leq  s^{1 - \alpha} (1 - s)^{\alpha} 
    \big((1 - \delta_{n}^{2} )^{1 - \alpha} F(\rho_{E|00}, \rho_{E|11})^{2n\alpha} + \delta_n^{\alpha} (1 + \delta_{n})^{1 -\alpha} F(\rho_{E|00}, \rho_{E|10})^{2n\alpha} \big).
\end{align*}

Since $F(\rho_{E|00}, \rho_{E|10})^2 < 1$, we can guarantee existence of a parameter $\alpha^* {\in (\frac{1}{2},1)}$ again by continuity, such that  
\begin{equation*}
\forall \alpha \in [\alpha^* , 1) : \quad F(\rho_{E|00}, \rho_{E|10} )^{2\alpha}  \leq   \left( \frac{\epsilon}{1 - \epsilon} \right)^{1 -\alpha} 
\end{equation*}
Substituting this bound, we obtain that for $\alpha \geq \alpha^*$,
\begin{align*}
   A_1 &\leq   s^{1 - \alpha} (1 - s)^{\alpha} 
    \left( F(\rho_{E|00}, \rho_{E|11})^{2n\alpha} + \left(\frac{\epsilon}{1 - \epsilon} \right)^n (1 + \delta_{n})^{1 - \alpha}\right) \\ 
    &\leq s^{1 - \alpha} (1 - s)^{\alpha} 
    \left( F(\rho_{E|00}, \rho_{E|11})^{2n\alpha} + 2 \left(\frac{\epsilon}{1 - \epsilon} \right)^n \right) 
\end{align*}
where the final inequality stems from the fact that $(1 + \delta_{n}) <  2$. Thus there exists $\alpha^*$ such that $\forall \alpha \in [\alpha^* , 1)$ and appropriately fixed $\beta$,
\begin{eqnarray*}
p_{\error}({C|E})_{{\rho'(s)}} \leq 160 s^{1- \alpha} (1 -s)^{\alpha} F(\rho_{E|00} , \rho_{E|11})^{2 n \alpha} + c_{\alpha,\beta}(s) \left(\frac{\epsilon}{1 - \epsilon} \right)^{n}
\end{eqnarray*}
where 
\begin{eqnarray*}
 c_{\alpha,\beta}(s) = 160 \left(s +  {(1+2^\beta)}  s^{1-\beta} (1 -s)^{ \beta} + 2s^{1 - \alpha} (1- s)^{\alpha}  \right). 
\end{eqnarray*}
For $\alpha < 1$, 
\begin{eqnarray*}
    \int_{0}^{1} \frac{\dd s}{s} s^{1 - \alpha} (1 - s)^{\alpha} = \frac{\pi \alpha}{\sin(\pi \alpha)} < \infty.  
\end{eqnarray*}
This proves the claim. 
\end{proof}
Finally, we prove the main result of this section.
\begin{corollary}
In the AD-QKD protocol, under the assumption that the post-measurement states $\rho_{E|ij}$ are pure and mutually distinct, then 
\begin{eqnarray*}
    \lim_{n \rightarrow \infty} \mathscr{R}_{n} = 0
\end{eqnarray*}
if $Q(\rho_{E|00} , \rho_{E|11}) <  \frac{\epsilon}{1 - \epsilon}$.
\end{corollary}
\begin{proof}
To show that $\lim_{n \rightarrow \infty} \mathscr{R}_{n} = 0$, following Lemma \ref{corr: neccessary_and_sufficiant_condition}, it suffices to show that
\begin{eqnarray*}
 \lim_{n \rightarrow \infty} \frac{1}{-n \left(\frac{\epsilon}{1 - \epsilon} \right)^n \log \left(\frac{\epsilon}{1 - \epsilon} \right) } \int_{0}^{1} \frac{\dd s}{2 s 
 \ln 2} \left( p_{\error}({C|E})_{{\bar{\rho}(s)}} + p_{\error}({C|E} )_{{\bar{\rho}(1-s)}}\right) =  0.
\end{eqnarray*}
Using Theorem \ref{thm:main_result_pure} and applying Lemma \ref{lemm: min-entropy_monotonicity}, we upper bound the error probabilities $  p_{\error}({C|E})_{{\bar{\rho}(s)}} $ and  $  p_{\error}({C|E})_{{\bar{\rho}(1-s)}} $

\begin{eqnarray*}
 \forall \alpha \in [ \alpha_1^* ,  1) : \quad    p_{\error}({C|E})_{{\bar{\rho}(s)}} &\leq& 160 s^{1 - \alpha} (1 -s)^{\alpha} F(\rho_{E|00} , \rho_{E|11})^{2 {n} \alpha} + c_1 (s) {\left(\frac{\epsilon}{1-\epsilon}\right)}^n  \\
 \forall \alpha \in [\alpha_2^* , 1):   p_{\error}({C|E})_{{\bar{\rho}(1-s)}} &\leq& 160s^{1 - \alpha} (1 -s)^{\alpha}  F(\rho_{E|{{11}}} , \rho_{E|{00}})^{2 {n}\alpha} + c_2(s) {\left(\frac{\epsilon}{1-\epsilon}\right)}^n.
\end{eqnarray*}

Here we note that relabelling which outcome the classical register $C$ calls $0$ and which it calls
$1$ does not change the error probability, and maps $\bar{\rho}(1-s)$ to
$s \proj{0} \ot \omega_{1}^{(n)} + (1-s) \proj{1} \ot \omega_{0}^{(n)}$. This is again of
the form of $\bar{\rho}(s)$, but with $\omega_{0}^{(n)}$ and $\omega_{1}^{(n)}$
interchanged, that is, with the conditional states permuted according to
$\rho_{E|ab} \mapsto \rho_{E|a \oplus 1, \, b \oplus 1}$. This permutation leaves the
hypotheses of Theorem~\ref{thm:main_result_pure} invariant, since it merely exchanges the
conditions $F(\rho_{E|00},\rho_{E|10}) < 1$ and $F(\rho_{E|11},\rho_{E|01}) < 1$. The
theorem therefore applies to the relabelled problem, returning the same bound with
$F(\rho_{E|00},\rho_{E|11})$ replaced by $F(\rho_{E|11},\rho_{E|00})$, possibly with a
different threshold $\alpha_{2}^{*}$.

Now noting that $F(\tau, \sigma)^2 = F(\sigma , \tau)^2$ we can choose $\alpha^* = \max \{ \alpha_1^{*} , \alpha_{2}^* \}$ to get
\begin{eqnarray*}
p_{\error}({C|E})_{{\bar{\rho}(s)}} +p_{\error}({C|E})_{{\bar{\rho}(1-s)}} {\leq}   {320}  s^{1 - \alpha} (1 - s)^{\alpha}  F(\rho_{E|00} , \rho_{E|11})^{2 n \alpha} {+2{c(s)}\left(\frac{\epsilon}{1-\epsilon}\right)^n},  
\end{eqnarray*}
where $\alpha $ is any real in the interval $[\alpha^* , 1)$ and $c(s):=\frac{1}{2}(c_1(s)+c_2(s))$. Thus, 
\begin{eqnarray*}
\lim_{n \rightarrow \infty} \mathscr{R}_{n} \leq {\lim_{n \to \infty}}\frac{1}{-n \left(\frac{\epsilon}{1 - \epsilon} \right)^n \log(\frac{\epsilon}{1- \epsilon})}  \int_{0}^{1} \frac{\dd s}{2 s \ln 2 } \left(320 {s^{1-\alpha}(1-s)^{\alpha}} F(\rho_{E|00} , \rho_{E|11})^{2 n \alpha}  + 2{c(s)} \left(\frac{\epsilon}{1- \epsilon} \right)^n \right)      
\end{eqnarray*}
Now, consider the second term inside the integral. We analyse its behaviour as $n \to \infty$. We have
\begin{eqnarray*}
 \frac{1}{-n \left( \frac{\epsilon}{1 - \epsilon} \right)^n \log\left(\frac{\epsilon}{1 - \epsilon} \right) \ln 2}  
 \int_{0}^{1} {\frac{ 2c(s)}{2s}} \left( \frac{\epsilon}{1 - \epsilon} \right)^n \dd s 
 = \frac{1}{-n \log\left(\frac{\epsilon}{1 - \epsilon} \right) \ln 2}  
 \int_{0}^{1} \frac{{c(s)}}{s}  \dd s 
 {=} \frac{1}{n} {k}, 
\end{eqnarray*}
where ${k}$ is a constant satisfying  ${k} {=} -\frac{1}{\log \left(\frac{\epsilon}{1-\epsilon} \right)} \int_{0}^{1}\dd s  \frac{{c(s)}}{s \ln 2} .$ Since the factor $\frac{1}{n}$ in front of ${k}$ vanishes as $n \to \infty$, this term goes to zero.

We now discuss the final term. By assumption we have $F(\rho_{E|00},\rho_{E|11})^{2} = Q(\rho_{E|00},\rho_{E|11})
< \frac{\epsilon}{1-\epsilon}$, so by continuity of $\alpha \mapsto x^{\alpha}$ there exists
$\alpha' \in [\alpha^{*},1)$ such that
\[
  \forall \alpha \in [\alpha',1) : \quad
  F(\rho_{E|00},\rho_{E|11})^{2\alpha} \leq \frac{\epsilon}{1-\epsilon} .
\]
Fixing any such $\alpha$, and using
$\int_{0}^{1} \frac{\dd s}{s} \, s^{1-\alpha} (1-s)^{\alpha} = \frac{\pi\alpha}{\sin(\pi\alpha)}$,
\begin{eqnarray*}
\frac{1}{-n \left(\frac{\epsilon}{1-\epsilon}\right)^{n}
  \log\left(\frac{\epsilon}{1-\epsilon}\right)}
\int_{0}^{1} \frac{\dd s}{2s \ln 2 \,} \,
320 \, s^{1-\alpha} (1-s)^{\alpha} F(\rho_{E|00},\rho_{E|11})^{2n\alpha}
= \frac{320 \pi \alpha}{2 n \ln 2 \, \sin(\pi\alpha)
  \log\left(\frac{1-\epsilon}{\epsilon}\right)}
\left( \frac{F(\rho_{E|00},\rho_{E|11})^{2\alpha}}{\frac{\epsilon}{1-\epsilon}} \right)^{\! n} .
\end{eqnarray*}
For any $p < \infty $ and $a \in (0, 1]$ we have that 
\begin{eqnarray*}
    \lim_{n \rightarrow \infty}\frac{p}{n} a^{n} = 0,
\end{eqnarray*}    
thus, the entire expression for $\mathscr{R}_{n}$ goes to zero as $n \to \infty$.

\end{proof}

We end this section with the remark that when Alice and Bob share a two-qubit state $\rho_{Q_A Q_B}$ and perform rank 1 projective measurements, then the state $\rho_{ET|ij}$ can be taken to be pure. This follows from the argument in the Supplementary Material of Ref. \cite{TanDI}. We repeat it here for the sake of completeness: 

Let $\rho_{Q_AQ_BE}=\ket{\psi}\!\bra{\psi}_{Q_AQ_BE}$ be the pure state shared before measurement. Since $\{M_{a|0}\}$, $\{N_{b|0}\}$ are projective measurements in the $x$-$z$ plane, the local Pauli $Y$ operator flips every outcome, $YM_{a|0}Y=M_{a\oplus1|0}$, and likewise for Bob. Define
\[
\ket{\psi'}_{Q_AQ_BET}=\frac{1}{\sqrt2}\Big(\ket{\psi}_{Q_AQ_BE}\otimes\ket{0}_T+(Y_A\otimes Y_B\otimes\id_E)\ket{\psi}_{Q_AQ_BE}\otimes\ket{1}_T\Big).
\]
Measuring $M_{a|0}\otimes N_{b|0}$ on $\ket{\psi'}_{Q_AQ_BET}$, for outcome $(a,b)$ gives a pure conditional state $\rho'_{ET|ab}$. The replacement does not alter the observed  statistics $p_{\tilde A\tilde B|00}(ab)$ and also preserves the fidelity in the sense that  
\[
F(\rho'_{ET|00},\rho'_{ET|11})=F(\rho_{ET|00},\rho_{ET|11}).
\]
Also, the conditional states remain mutually distinct (if they were mutually distinct to start with). We may therefore work with the pure states $\rho'_{ET|ab}$ throughout, and drop the prime henceforth.
\subsection{Method to bound two-qubit fidelities in the AD-DIQKD protocol}

As shown by \cite{tan2020advantage}, for the DIQKD protocol, if the abort condition is a function of the observed CHSH score, then it suffices for Alice and Bob to share a convex combination of two-qubit states $\rho_{Q_A Q_B}$. Here we prove that without any loss of generality, we can reduce to Bell diagonal states when doing the security analysis of DIQKD protocol. 

\begin{lemma}\label{lemm: reduction to Bell diagonal states}
For the AD-DIQKD protocol, if the abort condition is only a function of CHSH inequality violation, it suffices to consider the case where Alice and Bob share (a convex combination of) Bell diagonal states, and their projective measurements are real and defined in the Bell basis.
\end{lemma}

\begin{proof}
The proof rests on the proof of Theorem 2 in \cite{tan2020advantage}, where it is shown that the symmetrisation step [i.e. the step in which the raw bits $A_{i}$ and $B_{i}$ are modified to $\tilde{A}_{i}$ and $\tilde{B}_{i}$ by generating a bit $T_{i}$] renders the states $\ket{\psi}_{ABE|\mu}$ (Tan et. al. use the notation $\ket{\psi}_{ABE|\lambda}$ instead of $\ket{\psi}_{ABE|\mu}$ used here) and the states $\ket{\psi'}_{ABET|\mu}$

\begin{eqnarray}
\ket{\psi'}_{ABET|\mu} &=& \frac{1}{\sqrt{2}} \left( \ket{\psi}_{ABE|\mu} \otimes \ket{0} + (Y \otimes Y \otimes \id) \ket{\psi}_{ABE|\mu} \otimes \ket{1} \right) \nonumber
\end{eqnarray}

operationally equivalent -- meaning that any strategy that involves pre-sharing the state $\ket{\psi}_{ABE|\mu}$ followed by locally measuring the state and then performing the symmetrisation step is identical to a strategy in which the state $\ket{\psi'}_{ABET|\mu}$ is used, in the sense that both strategies produce the same experimental statistics and give identical (pessimistic) key rates. Here $Y$ is the Pauli $Y$ operator with respect to any basis in which the measurements $M_{a|x} := \proj{\varphi_{a|x}}$ and $N_{b|y} := \proj{\zeta_{b|y}}$ are real and take the form:

\begin{eqnarray}\label{eqn: real_measurements}
    \ket{\varphi_{a|x}} =  \cos(\varphi_{a|x}) \ket{0}_{A} + \sin (\varphi_{a|x}) \ket{1}_{A}, \quad \ket{\zeta_{b|y}} = \cos(\zeta_{b|y}) \ket{0}_{B} + \sin(\zeta_{b|y}) \ket{1}_{B}.
\end{eqnarray}

Here, we stress that while we have used a degree of freedom to define the $x$-$z$ plane in the Bloch sphere representation on the state space $\mathcal{S}(\mathcal{H}_{A})$ and $\mathcal{S}(\mathcal{H}_{B})$ respectively, we still have a degree of freedom to choose the $x$-axis -- i.e., there is a degree of freedom to choose the Pauli $x$ and $z$ operators respectively.

To complete the proof of the lemma, we simply extend the above argument to make further reductions following \cite{BhRC, rutvij_thesis}. In particular, we argue that for the proof of Theorem 2 in \cite{tan2020advantage}, it suffices to consider the state $\ket{\psi''}_{ABET|\mu}$ instead of $\ket{\psi'}_{ABET|\mu}$, given by 

\begin{eqnarray}
    \ket{\psi''}_{ABET|\mu} = \frac{1}{2} \left(  \left( \ket{\psi}_{ABE|\mu}  +  (\ket{\psi}_{ABE|\mu})^{P} \right) \otimes \ket{0} + \left( (Y \otimes Y \otimes \id) \ket{\psi}_{ABE|\mu} \otimes \ket{1} + ((Y \otimes Y \otimes \id) \ket{\psi}_{ABE|\mu})^{P} \ot \ket{1} \right) \right) \nonumber
\end{eqnarray}

Here $(.)^{P}$ corresponds to taking the partial transpose over $\mathcal{H}_{A} \otimes \mathcal{H}_{B}$ induced by any Bell basis for which the measurements $M_{a|x}$ and $N_{b|y}$ are real -- as in, they take the form given by Eq.\eqref{eqn: real_measurements} (See Lemma 13 of \cite{BhRC} for a precise definition of $(.)^{P}$). Here, we take $\ket{\psi}_{ABE}$ and $(\ket{\psi_{ABE}})^{P}$ to be orthogonal states by construction by ensuring that they have orthogonal support in the space $E$.  To see this, we note from Lemma 13 of \cite{BhRC} that the post-measurement states obey the following:

\begin{eqnarray}
    \tr_{AB}(M_{a|x} \otimes N_{b|y} \otimes \id_{ET} \proj{\psi'}_{ABET|\mu}) = \tr_{AB}(M_{a|x} \otimes N_{b|y} \otimes \id_{ET} \proj{\psi''}_{ABET|\mu})  
\end{eqnarray}

if $M_{a|x}$ and $N_{b|y}$ are rank-one projectors. Thus, $\ket{\psi'}$ can be replaced by $\ket{\psi''}$ in the proof of Theorem 2 in \cite{tan2020advantage}, as the post-measurement states held by Eve and the statistics generated by measuring the two states are identical.

The difference in using $\ket{\psi'}_{ABET|\mu}$ and $\ket{\psi''}_{ABET|\mu}$ lies in the marginal states $\rho'_{AB|\mu} = \tr_{ET} \proj{\psi'}_{ABET|\mu}$ and $\rho''_{AB|\mu} = \tr_{ET} (\proj{\psi''}_{ABET|\mu})$. The main difference is that $\rho''_{AB|\mu} = (\rho''_{AB|\mu})^{T}$, whereas such a relation does not hold for $\rho'_{AB|\mu}$. In other words, the state $\rho''_{AB|\mu}$ can be taken to be real in the family of Bell bases in which the measurement operators are real.

The final step is to choose this Bell basis such that the state $\rho_{AB}''$ is Bell diagonal. This is again shown in \cite{PABGMS} (also stated more formally in Lemma 14 of \cite{BhRC}), where it is shown that for every $\rho_{AB}''$ satisfying $(\rho_{AB}'')^{T} = \rho''_{AB}$ and $Y \otimes Y \rho''_{AB} Y \otimes Y = \rho''_{AB}$, there exist local unitaries $U_{A}$ and $U_{B}$ that diagonalise $\rho_{AB}''$, while also satisfying $U_{A} Y  U_{A}^{\dagger} = Y$ and $U_{B} Y U_{B}^{\dagger} = Y$. In other words, $\rho''_{AB}$ can be taken to be Bell diagonal by freely choosing the $x$-axis locally.  

Furthermore, we can simultaneously ensure that the Bell diagonal state $\rho_{AB}'' = \lambda_{\Phi^+} \proj{\Phi^+ } + \lambda_{\Phi^-} \proj{\Phi^-} + \lambda_{\Psi^-} \proj{\Psi^-} + \lambda_{\Psi^+} \proj{\Psi^+}$ satisfies $\lambda_{\Phi^+} \geq \lambda_{\Phi^-} $ and $ \lambda_{\Phi^-} \geq \lambda_{\Psi^-}$ as was also done in \cite{PABGMS}.
\end{proof}

We now bound the fidelities
\(
    F(\rho_{E|ab}, \rho_{E|a'b'})
\)
between Eve's post–measurement states when Alice and Bob share a Bell–diagonal state. Let $x^*, y^*$ be the key generating settings. Since $M_{a|x^*}$ and $N_{b|y^*}$ are rank–one projectors, the conditional state held by Eve after obtaining outcomes $(a,b)$ is pure. More precisely,
\[
    \rho_{E|ab} = \proj{\psi_{E|ab}}, \qquad
    \sqrt{p(ab|x^*y^*)}\ket{\psi_{E|ab}}
    =\sum_{i}\sqrt{\lambda_i}\,
    \braket{\varphi_{a|x^*},\zeta_{b|y^*}}{i}\ket{i}_E,
\]
where $\{\ket{i}\}$ denotes the Bell basis, the $\lambda_i$ are the Bell–diagonal eigenvalues of $\rho_{AB}$, and
\[
    p(ab|x^*y^*)
    =\tr\!\big[(M_{a|x^*}\otimes N_{b|y^*})\rho_{AB}\big]
\]
is directly observable (e.g., determined by the noise parameter~$\epsilon$).

Because both conditional states are pure, the fidelity reduces to the squared overlap:
\[
    F(\rho_{E|ab},\rho_{E|a'b'})^2
    =\big|\braket{\psi_{E|ab}}{\psi_{E|a'b'}}\big|^2
    =\frac{\big|\braket{\tilde\psi_{E|ab}}
                      {\tilde\psi_{E|a'b'}}\big|^2}
           {p(ab|x^*y^*)\,p(a'b'|x^*y^*)},
\]
where $\ket{\tilde\psi_{E|ab}}$ is the unnormalised post–measurement vector.  
Substituting the expression above gives
\begin{align*}
    F(\rho_{E|ab},\rho_{E|a'b'})^2  
    &= \frac{\Big|\sum_i \lambda_i\,
       \braket{\varphi_{a|x^*},\zeta_{b|y^*}}{i}\,
       \braket{i}{\varphi_{a'|x^*},\zeta_{b'|y^*}}
       \Big|^2}
       {p(ab|x^*y^*)\,p(a'b'|x^*y^*)} \\
    &= \frac{
        \big|
            \bra{\varphi_{a|x^*},\zeta_{b|y^*}}
            \rho_{AB}
            \ket{\varphi_{a'|x^*},\zeta_{b'|y^*}}
        \big|^2
       }
       {p(ab|x^*y^*)\,p(a'b'|x^*y^*)}.
\end{align*}

Since $\rho_{AB}$ is Bell–diagonal, the fidelity expression may be bounded explicitly by parametrising $\rho_{AB}$ in terms of its eigenvalues $\{\lambda_i\}$.  
For obtaining tight key–rate bounds in the device-independent setting, one may further convert the optimisation over $\{\varphi_{a|x} , \zeta_{b|y}, \lambda_i\}$ and the resulting conditional probabilities into a commuting polynomial optimisation problem and apply SDP-based relaxations (see \cite{rutvij_thesis,sharma2025enhancing} for a more detailed explanation of this technique, as well as illustrative examples of its use). 

\section{Lower bounds on the von Neumann entropy in finite blocklength regime}\label{app: lower bounds finite}

To derive lower bounds on key rates in the finite blocklength regime, we begin by reviewing (and reproving for completeness) the standard technique that relates a quantum hypothesis testing problem to a classical one. We then solve the resulting classical hypothesis problems in order to get bounds on the von Neumann entropy. As we demonstrate in Fig. 2 of the main text, our method yields improved von Neumann entropy lower bounds in the finite blocklength regime.

For the upcoming sections, we adopt the following notation for convenience: if $P$ is a distribution on the finite set 
$\mathcal{X}$, then $P^{\otimes n}$ denotes the i.i.d.\ product distribution on $\mathcal{X}^n$, 
defined by
\[
P^{\otimes n}(x_1, x_2, \ldots, x_n) = P(x_1) P(x_2) \cdots P(x_n).
\]

\subsection{Lower bounds on the quantum hypothesis testing problem by a classical hypothesis testing problem} 

This section presents the details of the key insight that allows us to obtain lower bounds on the error probabilities for finite blocklength. In particular, we discuss how a quantum binary hypothesis testing problem can be lower bounded by a suitably chosen classical hypothesis testing problem. This is the central result of~\cite{nussbaum2009chernoff} (see also~\cite{audenaert2012quantum} for further discussion), which was used to show that the quantum Chernoff bound gives the correct error exponent in the asymptotic (binary) hypothesis testing problem. Here we use this reduction to compute the finite blocklength error probabilities and hence compute the entropy bounds, rather than to derive an exponent.

For completeness, we reproduce the main steps of~\cite{nussbaum2009chernoff}, adapting them to our setting by keeping the prior probabilities $\pi_0$ and $\pi_1$ general instead of fixing them to $\frac{1}{2}$ and $\frac{1}{2}$, as the role of the priors is crucial for our analysis.
 
\begin{lemma}[Theorem~2 from~\cite{nussbaum2009chernoff}]
\label{lemm: classical_hypothesis_testing_chernoff}
Let $\tau_0$ and $\tau_1$ be $d$-dimensional quantum states with spectral decompositions  
\[
\tau_0 = \sum_{i=1}^{d} \lambda_{\tau_0^{(i)}} \proj{{x_i}}, 
\quad
\tau_1 = \sum_{j=1}^{d} \lambda_{\tau_1^{(j)}} \proj{{y_j}}.
\]
for some orthonormal basis $\{x_i\}_i$ and $\{y_j\}_j$. Define the joint probability distributions
\[
P_0(i,j) := \lambda_{\tau_0^{(i)}} \,|\braket{{x_i}}{{y_j}}|^2, 
\qquad
P_1(i,j) := \lambda_{\tau_1^{(j)}} \,|\braket{{x_i}}{{y_j}}|^2.
\] 
Let
\[
\hat{\rho}^{(n)}(\pi_0) := \pi_0 \proj{0} \otimes \tau_0^{\ot n} 
+ \pi_{1} \proj{1} \otimes \tau_1^{\ot n}.
\]
Then the quantum minimal error probability satisfies
\[
p_{\error}(C|E)_{\hat{\rho}^{(n)}(\pi_0)} 
\ \geq\ \frac{1}{2} \,
p_{\error}^{C}\big( \{ \pi_0\,P_0^{\otimes n},\ \pi_1\,P_1^{\otimes n} \} \big).
\]
Here $p_{\error}^{C}( .)$ denotes the optimal error probability of distinguishing a classical hypothesis  -- i.e.,  
\(p_{\error}(C|E)\) evaluated for the cq state  
\(\pi_{0} \proj{0} \otimes P^{\ot n}_0 + \pi_{1} \proj{1} \otimes P_1^{\otimes n}\),  
where $P_0$ and $P_1$ are regarded as diagonal density operators.
\end{lemma}

\begin{proof}
We start by proving this for $n =1$.
By definition of $p_{\error}(C|E)_{\hat{\rho}^{(1)}(\pi_0)}$ we have that
\begin{align*}
p_{\error}(C|E)_{\hat{\rho}^{(1)}(\pi_0)} &= \pi_0 \tr(\Pi \tau_0) + \pi_1 \tr((\id - \Pi)\tau_1) \\
&= \sum_{i,j=1}^{d} \left( \pi_0 \lambda_{\tau_0^{(i)}} |\bra{{x_i}} \Pi \ket{{y_j}}|^2 + \pi_1 \lambda_{\tau_1^{(j)}} |\bra{{x_i}} (\id - \Pi) \ket{{y_j}}|^2 \right)
\end{align*}
where $\Pi$ is the Helstrom projector (corresponding to $\tau_0$). Using the inequality $|a|^2 + |b|^2 \geq \frac{1}{2}|a + b|^2$ for any $a,b \in \mathbb{C}$, we obtain
\begin{align}
p_{\error}(C|E)_{\hat{\rho}} &\geq \frac{1}{2} \sum_{i,j=1}^{d} \min \left\{ \pi_0 \lambda_{\tau_0^{(i)}}, \pi_1 \lambda_{\tau_1^{(j)}} \right\} |\braket{{x_i}}{{y_j}}|^2.
\end{align}
Now, observe that
\[
\min\left\{ \pi_0 P_0(i,j), \pi_1 {P}_1(i,j) \right\} = \min \left\{ \pi_0 \lambda_{\tau_0^{(i)}}, \pi_1 \lambda_{\tau_1^{(j)}} \right\} |\braket{{x_i}}{{y_j}}|^2.
\]
Hence,
\[
p_{\error}(C|E)_{\hat{\rho}^{(1)}(\pi_0)} \geq \frac{1}{2} \sum_{i,j=1}^{d} \min\left\{ \pi_0 P_0(i,j), \pi_1 P_1(i,j) \right\} = \frac{1}{2} p_{\error}^{C}\big( \{ \pi_0\,P_0,\ \pi_1\,P_1 \} \big).
\]
Substituting $\tau_0 \mapsto \tau_0^{\ot n}$ , $\tau_1 \mapsto \tau_1^{\ot n}$ also implies $P_0 \mapsto P_0^{\ot n}$ and ${P}_1\mapsto{P}_1^{\ot n}$ (since the spectral decomposition of $\tau^{\ot n}$ is $n$ fold product of $\tau$ and therefore, the induced classical distributions are also of the form $P^{\ot n}$). This concludes the proof for all $n \in \mathbb{N}$.
\end{proof} 

Lemma \ref{lemm: classical_hypothesis_testing_chernoff} reduces the problem of computing the error probability in quantum hypothesis testing to a classical hypothesis testing problem. This reduction simplifies the analysis considerably, as it essentially involves computing the trace distance of two diagonal matrices. When the quantum states \(\tau_0\) and \(\tau_1\) are known, and for any finite \(n\), this classical hypothesis testing can be performed numerically by first carrying out a spectral decomposition of the states, which allows the definition of classical distributions \(P_0\) and \(P_1\). This method scales as \(d^{2n}\), where \(d\) is the dimension of the Hilbert space, since \(P_0\) and \(P_1\) are supported on a classical random variable with \(d^2\) outcomes -- compared to \(d^{3n}\) for the full quantum hypothesis testing problem.

\subsection{States used to generate Fig. 2 in the main text}\label{app: figure_data}
To generate Fig. $2$ we considered the states 
\begin{eqnarray*}
    \omega_{0} =
    \begin{pmatrix}
0.0718
& -0.0870 + 0.2091i
& -0.0039 - 0.0064i
& 0.0127 - 0.0027i
\\
-0.0870 - 0.2091i
& 0.9177
& -0.0234 + 0.0383i
& -0.0039 - 0.0519i
\\
-0.0039 + 0.0064i
& -0.0234 - 0.0383i
& 0.0043
& -0.0034 + 0.0008i
\\
0.0127 + 0.0027i
& -0.0039 + 0.0519i
& -0.0034 - 0.0008i
& 0.0062
\end{pmatrix}
\end{eqnarray*}
and
\begin{eqnarray*}
    \omega_{1} =
    \begin{pmatrix}
0.6927
& 0.0191 + 0.0361i
& 0.0955 + 0.2505i
& -0.0276 - 0.3148i
\\
0.0191 - 0.0361i
& 0.0542
& 0.0130 + 0.0077i
& -0.0126 - 0.0107i
\\
0.0955 - 0.2505i
& 0.0130 - 0.0077i
& 0.1076
& -0.1190 - 0.0326i
\\
-0.0276 + 0.3148i
& -0.0126 + 0.0107i
& -0.1190 + 0.0326i
& 0.1455
\end{pmatrix}.
\end{eqnarray*}
It can be verified using explicit computation that $F(\omega_0 , \omega_1) \approx  0.416$.  The resulting classical hypothesis testing problem would be to distinguish the classical hypotheses $P_0^{\ot n}$ and $P_{1}^{\ot n}$ where $P_0$ and $P_1$ are given by
\[
\begin{aligned}
P_{0} \approx {}&
(0.0000,\, 0.0000,\, 0.0000,\, 0.0005,\, 0.0000,\, 0.0014,\, 0.0005,\, 0.0022,\\
&\quad 0.0000,\, 0.0000,\, 0.0010,\, 0.9282,\, 0.0000,\, 0.0000,\, 0.0174,\, 0.0487),
\\[1ex]
P_{1} \approx {}&
(0.0001,\, 0.0000,\, 0.0000,\, 0.0000,\, 0.0000,\, 0.0031,\, 0.0001,\, 0.0000,\\
&\quad 0.0000,\, 0.0000,\, 0.0027,\, 0.0501,\, 0.0006,\, 0.0260,\, 0.8703,\, 0.0469).
\end{aligned}
\]
where the entries of $P_0$ and $P_1$ are vectorized in column-major lexicographical order, i.e., with the row index varying first and the column index varying second.

We can compute the classical error probability by computing the total variation distance TVD of the hypothesis $P_0^{\ot n}$ and $P_1^{\ot n}$ using exact methods. The cost of evaluating the TVD for large $n$ can be reduced by performing the calculation using FFT (Fast Fourier Transform) based techniques.  
\section{Supplementary mathematical results}\label{app: maths} 
In this section, we present all the supplementary mathematical results that go towards proving other results in the different sections of the Supplementary Material. 

\subsection{Results concerning the binary entropy}

We begin by proving an integral representation of the binary entropy, which will be instrumental in establishing Lemma~\ref{lemm: known_lower_bounds_using_the_integral_representation}. The following result provides an alternative expression for the binary entropy in terms of an integral, which may be of independent interest. For related formulations and applications in deriving device-independent randomness and key rates, see also~\cite{rutvij_thesis}.

\begin{lemma}\label{lem:integral_expression_binary_entropy}
The binary entropy admits the following integral representation:
\begin{equation}
    \forall x \in [-1,1]: \quad h_2\left( \frac{1+x}{2} \right) = 1 - \int_0^1 \frac{\dd s}{2s \ln 2} \left( 1 - \sqrt{1 - 4s(1-s)x^2} \right).
\end{equation}
\end{lemma}

\begin{proof}
We prove the identity by computing the power series expansion of both sides and showing they agree term-by-term.

\paragraph{Left-hand side: Series expansion of the binary entropy.} Recall that for \( x \in [-1,1] \), the binary entropy satisfies:
\[
h_2\left( \frac{1+x}{2} \right) = -\frac{1+x}{2} \log\left( \frac{1+x}{2} \right) - \frac{1-x}{2} \log\left( \frac{1-x}{2} \right).
\]
Expanding \( \log(1 \pm x) \) using the Taylor series
\[
\log(1 \pm x) = \pm \sum_{n=1}^{\infty} \frac{(\mp 1)^{n +1} x^n}{n \ln 2},
\]
and simplifying, one obtains the standard power series expansion:
\begin{equation}
    h_2\left( \frac{1+x}{2} \right) = 1 - \frac{1}{\ln 2} \sum_{n=1}^{\infty} \frac{x^{2n}}{2n(2n - 1)}.
\end{equation}

\paragraph{Right-hand side: Series expansion of the integral.} Now consider the right-hand side of the proposed identity. We expand the square root in a power series:
\[
\sqrt{1 - 4s(1-s)x^2} = \sum_{n=0}^{\infty} \binom{1/2}{n} \left[ -4s(1-s)x^2 \right]^n = \sum_{n=0}^{\infty} \binom{1/2}{n} (-1)^n 4^n x^{2n} s^n (1-s)^n.
\]
Subtracting from 1 and integrating term-by-term, we get:
\begin{align*}
     - \int_0^1 \frac{\dd s}{2s \ln 2} \left( 1 - \sqrt{1 - 4s(1-s)x^2} \right)  
    &= -\sum_{n=1}^{\infty} \frac{x^{2n}}{2 \ln 2} \binom{1/2}{n} (-1)^{n{+1}} 4^n \int_0^1 s^{n-1} (1-s)^n \dd s.
\end{align*}
Using the Beta function identity
\[
\int_0^1 s^{n-1} (1-s)^n \dd s = B(n, n+1) = \frac{\Gamma(n)\Gamma(n+1)}{\Gamma(2n+1)} = \frac{n!(n-1)!}{(2n)!},
\]
we obtain:
\begin{equation}
    - \int_0^1 \frac{\dd s}{2s \ln 2} \left( 1 - \sqrt{1 - 4s(1-s)x^2} \right) = \sum_{n=1}^{\infty} \frac{x^{2n}}{2 \ln 2} \binom{1/2}{n} (-1)^{{n}} 4^n \cdot \frac{n! (n-1)!}{(2n)!}.
\end{equation}
We now simplify the expression further. {First note that for \(n=1\), the coefficient of \(x^{2n}\) is \(-\frac{1}{2\ln2}\)}. Recall that
\[
\binom{1/2}{n} = \frac{(-1)^{n+1} (2n - 3)!!}{2^n n!}, \quad \text{for } n \geq {2}.
\]
The above identity allows us to simplify the coefficient of $x^{2n}$ in the series expansion as follows
\begin{align*}
    {-}\frac{1}{2 \ln 2} \cdot \frac{(2n - 3)!!}{2^n n!} \cdot 4^n \cdot \frac{n! (n-1)!}{(2n)!} 
    &= {-}\frac{1}{2 \ln 2} \cdot \frac{(2n - 3)!!}{2^n} \cdot 2^{2n} \cdot \frac{(n-1)!}{(2n)!} \\
    &= {-}\frac{1}{2 \ln 2} \cdot \frac{(2n - 3)!! \cdot 2^n \cdot (n-1)!}{(2n)!}.
\end{align*}
Further observe that $(2n - 3)!! \cdot 2^n \cdot (n-1)! = {2(2n-2)!}$, implying
\[
\frac{(2n - 3)!! \cdot 2^n \cdot (n-1)!}{(2n)!} = \frac{1}{n(2n - 1)}.
\]
Hence, the series expansion for the right-hand side of becomes the identity:
\[
{-}\sum_{n=1}^{\infty} \frac{x^{2n}}{2n(2n - 1)\ln 2},
\]
which agrees exactly with the expansion of the binary entropy. This completes the proof.
\end{proof}

Next, we analyse the behaviour of the binary entropy term \( h_2(\delta_n) \) in the asymptotic limit. Specifically, we show that it scales as
\[
- n \log\left(\frac{\epsilon}{1 - \epsilon} \right) \cdot \left( \frac{\epsilon}{1 - \epsilon} \right)^{n}.
\]
for large $n$.

\begin{lemma}\label{lemm: ratio}
Let \(\delta_n := \frac{\epsilon^{n}}{\epsilon^{n} + (1 - \epsilon)^n}\), where \(\epsilon \in (0, \tfrac{1}{2})\). Then,
\begin{eqnarray}
    \lim_{n \rightarrow \infty} \frac{h_2(\delta_n)}{ - n \left( \frac{\epsilon}{1 - \epsilon} \right)^n \log \left( \frac{\epsilon}{1 - \epsilon} \right) } = 1.
\end{eqnarray}
\end{lemma}

\begin{proof}
First, we compute the asymptotic behaviour of the binary entropy function $h_2(\delta_n)$. Using L'Hôpital's rule twice, we obtain  
\begin{equation}
    \lim_{n \rightarrow \infty }\frac{-\delta_n \log(\delta_n)}{h_2(\delta_n)} 
    = \lim_{n \rightarrow \infty} \frac{1 + \ln(\delta_n)}{\ln(\delta_n) - \ln(1 - \delta_n)} 
    = \lim_{n \rightarrow \infty} (1 - \delta_n) = 1.
\end{equation}
Furthermore, we evaluate the ratio  
\begin{equation}
    \lim_{n \rightarrow \infty} \frac{\left(\frac{\epsilon}{1 - \epsilon}\right)^n}{\delta_n} 
    = \lim_{n \rightarrow \infty} \frac{\epsilon^n + (1 - \epsilon)^n}{(1 - \epsilon)^n} = 1.
\end{equation}
Combining these results, we find  
\begin{align}
    \lim_{n \rightarrow \infty} \frac{-n \left( \frac{\epsilon}{1 - \epsilon} \right) ^{n}\log\left( \frac{\epsilon}{1 - \epsilon} \right) }{h_2(\delta_n)} 
    &= \lim_{n\rightarrow\infty} \left( \frac{-\delta_n \log(\delta_n)}{h_2(\delta_n)} \right)  
    \lim_{n \rightarrow \infty} \left(\frac{\left( \frac{\epsilon}{1 - \epsilon} \right)^n \log\left(\left( \frac{\epsilon}{1 - \epsilon} \right)^{n}  \right) }{\delta_n \log(\delta_n)} \right)   \\ 
    &= \lim_{n \rightarrow \infty} -\frac{(1-\epsilon )^{-n} \left(\epsilon^n+(1-\epsilon )^n\right) \ln
   \left(\frac{\epsilon }{1-\epsilon }\right)}{2 \tanh ^{-1}(1-2 \epsilon )} \\
    &= 1.
\end{align}
Again, in the last equality, we have used L'Hôpital's rule. This completes the proof.
\end{proof}
\subsection{Continuity result}

We present a continuity bound for the sandwiched R\'enyi quantity $\tilde{Q}_{\alpha}(\cdot \| \cdot)$. This bound is essential for proving Lemma~\ref{thm: main_result_general}, where it allows us to upper-bound the error probabilities in terms of the individual sandwiched R\'enyi quantities $\tilde{Q}_{\alpha}(\rho_{E|ij}^{\otimes n} \| \rho_{E|kl}^{\otimes n})$, instead of the full quantity $\tilde{Q}_{\alpha}(\omega_{i}^{(n)} \| \rho_{E|kl}^{\otimes n})$.

\begin{lemma}\label{lemm: continuity_of_tilde_Q}
Let $\{ \rho_{E|ij} \}_{i,j}$ be some normalised states. Define states $\omega_{i}^{(n)} = (1 - \delta_{n}) \rho_{E|ii}^{\otimes n} + \delta_{n} \rho_{E|i\bar{i}}^{\otimes n}$ for some $\delta_{n} \in (0 ,1)$. 
Then for \( \alpha \in (1/2,1) \) and normalised state $\tau$, we have:
\begin{equation}
\tilde{Q}_{\alpha}(\omega_{i}^{(n)} \| \tau^{\ot n}) 
\leq 
\frac{
    \tilde{Q}_{\alpha}(\rho_{E|ii} \| \tau)^n + \delta_n^{\alpha} \tilde{Q}_{\alpha}(\rho_{E|i\bar{i}} \| \tau)^n
}{
    (1 + \delta_n)^{\alpha - 1}
}
- 
\delta_n \tilde{Q}_{\alpha}(\rho_{E|ii} \| \tau)^n.
\end{equation} 
\end{lemma}

\begin{proof}
We prove the case for $i = 1$, the proof for $i = 0$ follows similarly.  We follow the same strategy as \cite{bluhm2024unified}. We start from the observation that 
\begin{eqnarray}
\omega_1^{(n)} + \delta_{n} \rho_{E|11}^{\ot n} =  \rho_{E|11}^{\ot n} + \delta_{n}\rho_{E|10}^{\ot n}.
\end{eqnarray}
Now consider the following: 
\begin{eqnarray*}
\tilde{Q}_{\alpha}(\omega_{1}^{(n) } + \delta_{n}  \rho_{E|11}^{\otimes n}  || \tau^{\ot n}) &=& (1 + \delta_{n})^{\alpha} \tilde{Q}_{\alpha}\left( \frac{1}{1 + \delta_{n}}\omega^{(n)}_{1}  + \frac{\delta_{n}}{1 + \delta_{n}}  \rho_{E|11}^{\ot n} || \frac{1}{1 + \delta_{n}} \tau^{\ot n} + \frac{\delta_{n}}{1 + \delta_{n}} \tau^{\ot n} \right)  \\ 
&\geq&  (1 + \delta_{n})^{\alpha -1} \tilde{Q}_{\alpha}(\omega_{1}^{(n)} || \tau^{\ot n}) + (1 + \delta_{n})^{\alpha -1} \delta_{n} \tilde{Q}_{\alpha}(\rho_{E| 11}^{\ot n} || \tau^{\ot n}) \\ 
&=& (1 + \delta_{n})^{\alpha -1} \tilde{Q}_{ \alpha}(\omega_1^{(n)} || \tau^{\ot n}) + (1 + \delta_{n})^{\alpha- 1} \delta_{n} \tilde{Q}_{\alpha}(\rho_{E| 11} || \tau)^{n}   
\end{eqnarray*}
where the inequality above follows from the joint concavity of $\tilde{Q}_{\alpha}$ for $\frac{1}{2}<\alpha < 1$~\cite{tomamichel2015quantum}.  \\ 
We also get the following upper bound using the subadditivity of $\tilde{Q}_{\alpha}(.||.)$~\cite{marwah2022uniform}:
\begin{eqnarray}
\tilde{Q}_{\alpha}(\rho_{E|11}^{\ot n} +  \delta_{n} \rho_{E|10}^{\ot n} || \tau^{\ot n} ) \leq \tilde{Q}_{\alpha}(\rho_{E|11}^{\ot n} || \tau^{\ot n}) + \delta_{n}^{\alpha} \tilde{Q}_{\alpha}(\rho_{E| 10}^{\ot n} || \tau^{\ot n}) = \tilde{Q}_{\alpha}(\rho_{E|11} || \tau)^{n} + \delta_n^\alpha\tilde{Q }_{\alpha}(\rho_{E|10} || \tau)^{n}.     
\end{eqnarray}
Combining the two inequalities we get 
\begin{eqnarray}
\tilde{Q}_{\alpha}(\omega_{1}^{(n)} || \tau^{\ot n}) \leq \frac{\tilde{Q}_{\alpha}(\rho_{E|11} || \tau)^{n} + \delta_{n}^{\alpha} \tilde{Q }_{\alpha}(\rho_{E|10} || \tau )^{n} }{(1 + \delta_{n})^{\alpha - 1}} - \delta_{n} \tilde{Q}_{\alpha}(\rho_{E|11} || \tau )^{n}.
\end{eqnarray}
\end{proof}

\subsection{Upper bounding the error probability}
 
In this section, we prove that the error probability of discriminating a two-state ensemble  
\[
\M{E}= \left\{ \left(p_{X}(0) + p_{X}(2) \right)  \left(\frac{p_{X}(0)}{p_{X}(0) + p_{X}(2)} \omega_0 + \frac{p_{X}(2)}{p_{X}(0) + p_{X}(2)} \omega_2 \right), p_{X}(1) \omega_1 \right\}
\]  
can be upper bounded by the error probability of distinguishing the three-state ensemble  
\[
\tilde{\M{E}} = \{ p_{X}(0) \omega_0 , p_{X}(1) \omega_1 , p_{X}(2) \omega_2 \}.
\]  
The proof is a fairly straightforward data-processing argument but serves as a crucial component of our main argument. We therefore provide the proof here for convenience.

\begin{lemma}\label{lemm: min-entropy_monotonicity}
Let $\rho =  p_{0} \proj{0} \ot {\omega}_{0} + p_1 \proj{1} \ot {\omega}_1 + p_2 \proj{2} \ot {\omega}_{2}$ and $\sigma =  (p_0 + p_2)\proj{0} \ot (\frac{p_{0}}{p_0 + p_2} {\omega}_{0} + \frac{p_2}{p_0  + p_2} {\omega}_2)  +  p_1 \proj{1} \ot {\omega}_1 $. Then 
\begin{eqnarray}
p_{\error}(X|E)_{\sigma} \leq p_{\error}(X |E)_{\rho}. 
\end{eqnarray}
\end{lemma}
\begin{proof} 
It suffices to show that $p_{\mathrm{guess}}(X|E)_{\sigma} \geq p_{\mathrm{guess}}(X|E)_{\rho}$. To see this consider the POVM $\{ M_0 ,  M_1 , M_2 \}$ that achieves optimal guessing probability $p_{\mathrm{guess}}(X|E)_{\rho}$: 
\begin{eqnarray*}
p_{\mathrm{guess}}(X|E)_{\rho} &=&  \sum_{i} p_i \tr(M_i {\omega}_{i}) = \tr(M_1p_1 {\omega}_1) + \tr(M_0 p_{0} {\omega}_{0} + M_2  p_2  {\omega}_{2}) \\ 
                               &\leq& \tr(M_1p_1 {\omega}_1) + \tr \left( \left( M_0 + M_2 \right)  (p_0 {\omega}_0 + p_{2} {\omega}_{2}) \right) \\ 
                               &\leq& p_{\mathrm{guess}}(X|E)_{\sigma}
\end{eqnarray*}
where the first inequality holds because the cross terms $\tr(M_{0} p_2 \omega_2) $ and $\tr(M_{2} p_0 \omega_0)$ are non-negative and the final inequality holds because $\{ M_{0} + M_{2}, M_1 \} $ is a valid POVM. 
\end{proof}

\bibliography{references}

\end{document}
\end{document}